\documentclass[sigconf,screen,authorversion]{acmart}
\usepackage{xspace,balance,tabularx,multirow}
\usepackage{flushend}
\usepackage{tikz}
\usepackage{pgfplots}
\pgfplotsset{compat=1.16}
\usetikzlibrary{patterns}
\usepackage{subfig}
\usepackage[ruled, vlined, linesnumbered]{algorithm2e}
\usepackage{xcolor}
\usepackage{colortbl}
\usepackage{bbold}
\SetKwComment{Comment}{$\triangleright$\ }{}
\usepackage{enumitem}
\usepackage{tablefootnote}
\usepackage{upgreek,textgreek}
\usepackage{pifont}
\usepackage[noabbrev]{cleveref}
\usepackage{titlecaps}
\usepackage{lipsum}

\pgfplotsset{every tick label/.append style={font=\tiny}}

\newlength{\starsize}
\newlength{\starspread}
\tikzset{starsize/.code={\setlength{\starsize}{#1}},
         starspread/.code={\setlength{\starspread}{#1}}}
\tikzset{starsize=1mm,
         starspread=3mm}
\pgfdeclarepatternformonly[\starspread,\starsize]
  {my fivepointed stars}
  {\pgfpointorigin}
  {\pgfqpoint{\starspread}{\starspread}}
  {\pgfqpoint{\starspread}{\starspread}}
  {
   \pgftransformshift{\pgfqpoint{\starsize}{\starsize}}
   \pgfpathmoveto{\pgfqpointpolar{18}{\starsize}}
   \pgfpathlineto{\pgfqpointpolar{162}{\starsize}}
   \pgfpathlineto{\pgfqpointpolar{306}{\starsize}}
   \pgfpathlineto{\pgfqpointpolar{90}{\starsize}}
   \pgfpathlineto{\pgfqpointpolar{234}{\starsize}}
   \pgfpathclose%
   \pgfusepath{fill}
  }

\makeatletter
\newcommand*\bigcdot{\mathpalette\bigcdot@{.5}}
\newcommand*\bigcdot@[2]{\mathbin{\vcenter{\hbox{\scalebox{#2}{$\m@th#1\bullet$}}}}}
\makeatother

\newcommand{\stitle}[1]{\vspace*{0.5em}\noindent{\bf #1.\/}}

\newcommand{\ie}{{\it i.e.},\xspace}

\newcommand{\U}{\mathcal{U}\xspace}
\newcommand{\V}{\mathcal{V}\xspace}

\newcommand{\G}{\mathcal{G}\xspace}
\newcommand{\N}{\mathcal{N}\xspace}
\newcommand{\C}{\mathcal{C}\xspace}
\newcommand{\EDG}{\mathcal{E}\xspace}

\newcommand{\OO}{\mathcal{O}\xspace}

\newcommand{\DM}{\mathbf{D}\xspace}
\newcommand{\IM}{\mathbf{I}\xspace}
\newcommand{\SM}{\mathbf{S}\xspace}
\newcommand{\MM}{\mathbf{M}\xspace}
\newcommand{\PM}{\mathbf{P}\xspace}

\newcommand{\YM}{\mathbf{Y}\xspace}
\newcommand{\XM}{\mathbf{X}\xspace}
\newcommand{\LM}{\mathbf{L}\xspace}

\newcommand{\HM}{\mathbf{H}\xspace}
\newcommand{\GM}{\mathbf{G}\xspace}

\newcommand{\FY}{\boldsymbol{\Upsilon}\xspace}
\newcommand{\SGVM}{\boldsymbol{\Sigma}\xspace}
\newcommand{\GaM}{\boldsymbol{\Gamma}\xspace}
\newcommand{\PsiM}{\boldsymbol{\Psi}\xspace}

\newcommand{\FM}{\mathbf{F}\xspace}
\newcommand{\BM}{\mathbf{B}\xspace}
\newcommand{\RM}{\mathbf{R}\xspace}
\newcommand{\TM}{\boldsymbol{\Phi}\xspace}
\newcommand{\ZM}{\mathbf{Z}\xspace}

\newcommand{\dimU}{d\xspace}

\newcommand{\QM}{\mathbf{Q}\xspace}

\newcommand{\algo}{\texttt{TPO}\xspace}

\newenvironment{customlegend}[1][]{%
    \begingroup
    \csname pgfplots@init@cleared@structures\endcsname
    \pgfplotsset{#1}%
}{%
    \csname pgfplots@createlegend\endcsname
    \endgroup
}%

\def\addlegendimage{\csname pgfplots@addlegendimage\endcsname}

\makeatletter
\newcommand\footnoteref[1]{\protected@xdef\@thefnmark{\ref{#1}}\@footnotemark}
\makeatother

\let\oldnl\nl
\newcommand{\nonl}{\renewcommand{\nl}{\let\nl\oldnl}}

\DeclareMathOperator{\Tr}{trace}

\SetKwComment{Comment}{/* }{ */}

\makeatletter 
\g@addto@macro{\@algocf@init}{\SetKwInOut{Parameter}{Parameters}} 
\makeatother

\definecolor{myred}{HTML}{fd7f6f}
\definecolor{myred_new}{HTML}{D8D8D8}
\definecolor{myred_new2}{HTML}{D7191C}
\definecolor{myblue}{HTML}{7eb0d5}
\definecolor{mygreen}{HTML}{b2e061}
\definecolor{mypurple}{HTML}{bd7ebe}
\definecolor{myorange}{HTML}{ffb55a}
\definecolor{myyellow}{HTML}{ffee65}
\definecolor{mypurple2}{HTML}{beb9db}
\definecolor{mypink}{HTML}{fdcce5}
\definecolor{mycyan}{HTML}{8bd3c7}

\definecolor{myblue2}{HTML}{115f9a}
\definecolor{myred2}{HTML}{c23728}

\apptocmd\normalsize{%
   \setlength{\abovedisplayskip}{1.5pt}
   \setlength{\belowdisplayskip}{1.5pt}
   \setlength{\abovedisplayshortskip}{1.5pt}
   \setlength{\belowdisplayshortskip}{1.5pt}
}{}{}

\newcommand{\eat}[1]{}

\AtBeginDocument{%
  \providecommand\BibTeX{{%
    \normalfont B\kern-0.5em{\scshape i\kern-0.25em b}\kern-0.8em\TeX}}}

\copyrightyear{2024}
\acmYear{2024}
\setcopyright{rightsretained}
\acmConference[KDD '24]{Proceedings of the 30th ACM SIGKDD Conference on Knowledge Discovery and Data Mining}{August 25--29, 2024}{Barcelona, Spain}
\acmBooktitle{Proceedings of the 30th ACM SIGKDD Conference on Knowledge Discovery and Data Mining (KDD '24), August 25--29, 2024, Barcelona, Spain}
\acmDOI{XXXXXXX.XXXXXXX}
\acmISBN{978-1-4503-XXXX-X/18/06}

\begin{document}

\title{Effective Clustering on Large Attributed Bipartite Graphs}
\subtitle{Technical Report}


\author{Renchi Yang}
\affiliation{%
  \institution{Hong Kong Baptist University}
  \country{}
}
\email{renchi@hkbu.edu.hk}
\orcid{0000-0002-7284-3096}

\author{Yidu Wu}
\authornote{Work done while at Hong Kong Baptist University.}
\affiliation{%
  \institution{Chinese University of Hong Kong}
  \country{}
}
\email{yiduwu@cuhk.edu.hk}

\author{Xiaoyang Lin}
\affiliation{%
  \institution{Hong Kong Baptist University}
  \country{}
}
\email{csxylin@hkbu.edu.hk}

\author{Qichen Wang}
\affiliation{%
  \institution{Hong Kong Baptist University}
  \country{}
}
\email{qcwang@hkbu.edu.hk}
\orcid{0000-0002-0959-5536}

\author{Tsz Nam Chan}
\affiliation{%
  \institution{Shenzhen University}
  \country{}
}
\email{edisonchan@szu.edu.cn}
\orcid{https://orcid.org/0000-0001-5851-7967}

\author{Jieming Shi}
\affiliation{%
  \institution{Hong Kong Polytechnic University}
  \country{}
}
\email{jieming.shi@polyu.edu.hk}
\orcid{0000-0002-0465-1551}

\renewcommand{\shortauthors}{}

\begin{abstract}
Attributed bipartite graphs (ABGs) are an expressive data model for describing the interactions between two sets of heterogeneous nodes that are associated with rich attributes, such as customer-product purchase networks and author-paper authorship graphs. Partitioning the target node set in such graphs into $k$ disjoint clusters (referred to as $k$-ABGC) finds widespread use in various domains, including social network analysis, recommendation systems, information retrieval, and bioinformatics. However, the majority of existing solutions towards $k$-ABGC either overlook attribute information or fail to capture bipartite graph structures accurately, engendering severely compromised result quality.
The severity of these issues are accentuated in real ABGs, which often encompass millions of nodes and a sheer volume of attribute data, rendering effective $k$-ABGC over such graphs highly challenging.



In this paper, we propose \algo, an effective and efficient approach to $k$-ABGC that achieves superb clustering performance on multiple real datasets. \algo obtains high clustering quality through two major contributions: (i) a novel formulation and transformation of the $k$-ABGC problem based on {\em multi-scale attribute affinity} specialized for capturing attribute affinities between nodes with the consideration of their multi-hop connections in ABGs, and (ii) a highly efficient solver that includes a suite of carefully-crafted optimizations for sidestepping explicit affinity matrix construction and facilitating faster convergence. Extensive experiments, comparing \algo against 19 baselines over 5 real ABGs, showcase the superior clustering quality of \algo measured against ground-truth labels. Moreover, compared to the state of the arts, \algo is often more than $40\times$ faster over both small and large ABGs.
\end{abstract}

\ccsdesc[300]{Mathematics of computing~Computations on matrices}
\ccsdesc[500]{Computing methodologies~Cluster analysis}
\ccsdesc[500]{Computing methodologies~Spectral methods}
\ccsdesc[500]{Information systems~Clustering}

\keywords{clustering, bipartite graphs, attributes, eigenvector}


\maketitle
\section{Introduction}
Bipartite graphs are an indispensable data structure used to model the interplay between two sets of entities from heterogeneous sources, e.g., author-publication associations, customer-merchant transactions, query-webpage pairing, and various user-item interactions on social media, e-commerce platforms, search engines, etc. In the real world, such graphs are often associated with rich attributes, e.g., the user profile in social networks, web page content in web graphs, hallmarks of pathways in cancer signaling networks, and paper keywords in academic graphs, which are termed {\em Attributed Bipartite Graphs} (hereinafter ABGs).

Given an ABG $\G$ with two disjoint node sets $\U$ and $\V$, $k$-{\em Attributed Bipartite Graph Clustering} ($k$-ABGC), a fundamental task of analyzing ABGs, seeks to partition the nodes in the node set of interest, e.g., $\U$ or $\V$, into $k$ non-overlapping clusters $\C_1,\C_2,\cdots,\C_k$, such that nodes within the same cluster $\C_i$ are close to each other in terms of both their attribute similarity and topological proximity in $\G$. Due to the omnipresence of ABGs, $k$-ABGC has seen a wide range of practical applications in social network analysis, recommender systems, information retrieval, and bioinformatics, such as user/content tagging \cite{pan2013automatic,yao2013annotation}, market basket analysis \cite{zha2001bipartite,zhang2021measuring}, document categorization \cite{dhillon2001co,wang2015incorporating}, identification of protein complexes, disease genes, and drug targets \cite{xu2016interactive,pavlopoulos2018bipartite}, and many others \cite{li2020hierarchical,xu2023effective,yan2018telecomm,ren2021ensemfdet,kim2020geosocial}.

As reviewed in Section \ref{sec:relatedwork}, existing solutions towards $k$-ABGC primarily rely on {\em bipartite graph co-clustering} (BGCC), {\em attributed graph clustering} (AGC), and {\em attributed network embedding} (ANE) techniques. Amid them, BGCC has been extensively investigated in the literature \cite{dhillon2001co,kluger2003spectral,ailem2015co,labiod2011co,dhillon2003information,xu2019deep} for clustering non-attributed bipartite graphs, whose basic idea is to simultaneously group nodes in $\U$ and $\V$ merely based on their interactions in $\G$, instead of clustering them severally. 
As pinpointed in prior works \cite{bothorel2015clustering}, the attributes present rich information to characterize the properties of nodes and hence, can complement scant topological information for better node clustering.  Consequently, BGCC methods exhibit subpar performance on ABGs as they overlook such information.



To leverage the complementary nature of graph topology and attributes for enhanced clustering effectiveness, considerable efforts \cite{yang2021effective,fanseu2023grace,cui2020adaptive,zhang2019attributed,lai2023re,bothorel2015clustering} have been invested in recent years towards devising effective AGC models and algorithms. Although these approaches enjoy improved performance over {\em unipartite} attributed graphs by fusing graph connectivity and attribute information of nodes via deep learning or sophisticated statistical models, they are sub-optimal for ABGs.


Over the past decade, network embedding has emerged as a popular and powerful tool for analyzing graph-structured data, especially those with nodal attributes. 
Notwithstanding a plethora of network embedding techniques invented \cite{cui2018survey,giamphy2023survey,yang2023pane}, most of them are designed for unipartite graphs. To capture the unique characteristics of bipartite graphs, \citet{huang2020biane} extend \textsf{node2vec} \cite{grover2016node2vec} to ABGs, at the expense of tremendous training overhead.
Adopting this category of approaches for $k$-ABGC requires a rear-mounted phase (e.g., $k$-Means) to cluster the node embeddings, which is not cost-effective given the high embedding dimensions (typically 128).
To summarize, existing approaches to $k$-ABGC either dilute clustering quality due to inadequate exploitation of attributes and bipartite graph topology, or incur vast computation costs, especially on sizable ABGs encompassing thousands of attributes, millions of nodes, and billions of edges.



In response to these challenges, we propose \algo, a novel \underline{T}hree-\underline{P}hase \underline{O}ptimization framework for $k$-ABGC that significantly advances the state of the art in $k$-ABGC, in terms of both result effectiveness and computation efficiency. First and foremost, \algo formulates the $k$-ABGC task as an optimization problem based on {\em multi-scale attribute affinity} (MSA), a new node affinity measure dedicated to ABGs. More concretely, the MSA of two homogeneous nodes $u_i,u_j$ in $\U$ of ABG $\G$ evaluates the similarity of their attributes aggregated from multi-hop neighborhoods, which effectively captures the affinity of nodes with consideration of both their attributes and topological connections in bipartite graphs.
However, calculating the MSA of all node pairs in $\G$ for clustering is prohibitively expensive for large graphs, as it entails colossal construction time and space consumption ($O(|\U|^2)$). On top of that, the exact optimization of our $k$-ABGC objective is also infeasible as an aftermath of its NP-hardness. 

To tackle these issues, \algo adopts a three-phase optimization scheme for an approximate solution with time and space costs linear to the size of $\G$. Under the hood, similar in spirit to kernel tricks \cite{liu2011kernel}, \algo first leverages a mathematical apparatus, {\em random features} \cite{rahimi2007random,yu2016orthogonal}, to bypass the materialization of the all-pairwise MSA. The clustering task is later framed as a non-negative matrix factorization, followed by a matrix approximation problem, based on our theoretically-grounded problem transformation. Particularly, the former attends to yielding an intermediate, while the latter iteratively refines the intermediary result to derive the eventual clusters. In addition to the linear-time iterative solvers, \algo further includes a greedy initialization trick for speeding up the convergence, and an attribute dimension reduction approach to conspicuously boost the practical efficiency of \algo over graphs with large attribute sets, without degrading result quality.
Our empirical studies, which involved 5 real ABGs and compared against 19 existing algorithms, demonstrate that \algo consistently attains superior or comparable clustering quality at a fraction of the cost compared to the state-of-the-art methods. For instance, on the largest Amazon dataset with over 10 million nodes and 22 million edges, \algo obtains the best clustering accuracy within 3 minutes, whereas the state-of-the-art demands more than 4 hours to terminate.







\section{Problem Formulation}\label{sec:preliminary}
\begin{table}[!t]
\centering
\renewcommand{\arraystretch}{1.0}
\begin{small}
\caption{Frequently used notations.}\vspace{-3mm} \label{tbl:notations}
\begin{tabular}{|p{0.44in}|p{2.36in}|}
    \hline
    {\bf Notation} &  {\bf Description}\\
    \hline
    $\U,\V,\EDG$   & The node sets $\U,\V$, and the edge set $\EDG$ of ABG $\G$.\\ \hline
    $\XM_\U,\XM_\V$   & Attribute vectors of nodes in $\U$ and $\V$.\\ \hline
    $d_\U,d_\V$   & Attribute dimensions of nodes in $\U$ and $\V$.\\ \hline
    $w(u_i,v_j)$ & Weight of edge $(u_i,v_j)$ in $\EDG$. \\ \hline
    $k$   & The number of clusters. \\ \hline
    $\alpha$   & Balance coefficient used in Eq. \eqref{eq:Z-obj}. \\ \hline
    $\gamma$   & Maximum number of iterations used in Eq. \eqref{eq:P}. \\ \hline
    $d$   & Dimension of ${\XM}^{\prime}_\U$ in Eq. \eqref{eq:Xd} ($d\le d_\U$). \\ \hline
    $T_f,T_g$   & Maximum number of iterations used in Algorithms \ref{alg:onmf} and \ref{alg:round}, respectively. \\ \hline
    $\LM_\U$   & Normalized adjacency matrix defined in Eq. \eqref{eq:L}. \\ \hline
    $\ZM_\U, \widehat{\ZM}_\U$   & Feature vectors of nodes in $\U$ and their normalized version defined in Eq. \eqref{eq:PX} and Eq. \eqref{eq:Z}, respectively. \\ \hline
    $s(u_i,u_j)$   & The MSA between nodes $u_i$ and $u_j$ defined in Eq. \eqref{eq:s}. \\ \hline
    $\RM$   & Matrix satisfies $\RM[i]\cdot\RM[j]\approx s(u_i,u_j)$. \\ \hline
    $\YM$   & The NCI matrix defined in Eq. \eqref{eq:Y}. \\ \hline
    $\FY$   & The continuous version of $\YM$ satisfying Eq. \eqref{eq:obj4}. \\ \hline
\end{tabular}%
\end{small}
\end{table}
\subsection{Notation and Terminology}\label{sec:notation}
We denote matrices using bold uppercase letters, e.g., $\MM\in \mathbb{R}^{n\times d}$, and the $i$-th row (resp. the $j$-th column) of $\MM$ is represented as $\MM[i]$ (resp. $\MM[:,j]$). Accordingly, $\MM[i,j]$ signifies the entry at the $i$-th row and $j$-th column of $\MM$. For each vector $\MM[i]$, we use $\|\MM[i]\|$ to represent its $L_2$ norm
and $\|\MM\|_F$ to represent the Frobenius norm of $\MM$.

Let $\G=(\U\cup \V, \EDG, \XM_{\U}, \XM_{\V})$ symbolize an {\em attributed bipartite graph} (ABG), where $\EDG$ is composed of edges connecting nodes in two disjoint node sets $\U$ and $\V$ and each edge $(u_i,v_j)$ is associated with an edge weight $w(u_i,v_j)$. 
Each node $u_i\in \U$ (resp. $v_i\in \V$) of $\G$ is characterized by a length-$d_\U$ (resp. length-$d_\V$) attribute vector $\XM_\U[i]$ (resp. $\XM_\V[i]$).
Further, we denote by $\BM_\U\in \mathbb{R}^{|\U|\times |\V|}$ the adjacency matrix of $\G$ from the perspective of $\U$, in which $\BM_\U[i,j]=w(u_i,v_j)$ if $(u_i,v_j)\in \EDG$ and 0 otherwise. 
Let $\DM_\U$ (resp. $\DM_\V$) be a $|\U|\times |\U|$ (resp. $|\V|\times |\V|$) diagonal matrix wherein the diagonal entry $\DM_\U[i,i]$ (resp. $\DM_\V[i,i]$) stands for the sum of the weights of edges incident to $u_i$ (resp. $v_i$), i.e., $\sum_{(u_i,v_\ell)\in \EDG}{w(u_i,v_\ell)}$ (resp. $\sum_{(u_\ell,v_i)\in \EDG}{w(u_\ell,v_i)}$). 
Table \ref{tbl:notations} lists the frequently used notations throughout the paper.

The overarching goal of $k$-ABGC is formalized in Definition \ref{def:kabgc} and exemplified in Figure \ref{fig:ABGC}. Note that by default, we regard $\U$ as the target node set to cluster. 
The number $k$ can be specified by users or configured by a preliminary procedure \cite{milligan1985examination}. 

\begin{figure}[!t]
\centering
\includegraphics[width=0.8\columnwidth]{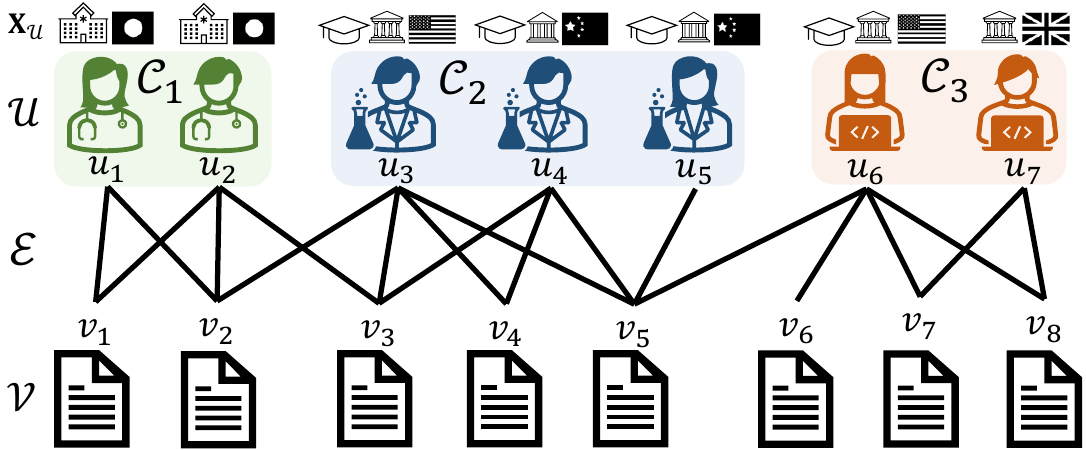}
\vspace{-3mm}
\caption{An Illustrative Example of $k$-ABGC}\label{fig:ABGC}
\end{figure}

\begin{definition}[$k$-Attributed Bipartite Graph Clustering ($k$-ABGC)]\label{def:kabgc} Given an ABG $\G$, the target node set $\U$, and the number $k$ of clusters, $k$-ABGC aims to partition node set $\U$ into $k$ disjoint clusters $\{\C_1,\C_2,\cdots,\C_k\}$ such that nodes within the same cluster are close to each other in terms of both topological proximity and attribute similarity while nodes across diverse clusters are distant.
\end{definition}

\eat{Figure \ref{fig:ABGC} exemplifies an ABG $\G$ with 7 researchers $u_1$-$u_7$ in $\U$, 8 research publications $v_1$-$v_8$ in $\V$, and the authorships in $\EDG$. Additionally, each researcher possesses a collection of attributes, including nationality, work institution, and academic qualifications. In Figure \ref{fig:ABGC}, $u_1,u_2$ share identical attributes and close collaboration, indicating they should be grouped in the same cluster. Similar observations can be made for node pairs $(u_3,u_4)$ and $(u_6,u_7)$, where the difference is that their attributes are partially analogous. Despite limited connections related to $u_5$, $u_5$ and $u_4$ are likely to be grouped together due to their identical attributes, with $u_4$ being the sole collaborator of $u_5$. Overall, given $k=3$, an intuitive and ideal solution for $k$-ABGC in Definition \ref{def:kabgc} is to divide the 7 researchers into 3 clusters, i.e., $\C_1=\{u_1,u_2\}$, $C_2=\{u_3,u_4,u_5\}$, and $\C_3=\{u_6,u_7\}$, with considering both their connectivity (collaboration) in $\G$ and attribute homogeneity. }

\subsection{Multi-Scale Attribute Affinity (MSA)}\label{sec:MSA}
Notice that Definition \ref{def:kabgc} cannot directly guide the generation of clusters, as it lacks a concrete optimization objective that quantifies node affinities. To this end, we first delineate our novel affinity measure MSA for nodes in terms of both graph structure and attributes, before formally introducing our objective in Section \ref{sec:obj}. 

\stitle{MSA formulation} We first assume that each node $u_i\in \U$ can be represented by a feature vector ${\ZM}_\U[i]$, which characterizes both the attributes as well as the rich semantics hidden in the bipartite graph topology.
Following the popular Skip-gram model \cite{mikolov2013distributed} and its extension to graphs \cite{perozzi2014deepwalk,grover2016node2vec}, we can model pair-wise affinity of nodes as a softmax unit \cite{goodfellow20166} parametrized by a dot product of their feature vectors. 
Rather than using the vanilla softmax function, we adopt a symmetric softmax function and formulate the MSA $s(u_i,u_j)$ between any two nodes $u_i,u_j$ in $\U$ as follows:
\begin{small}
\begin{equation}\label{eq:s}
s(u_i,u_j) = \frac{e^{\widehat{\ZM}_\U[i]\cdot \widehat{\ZM}_\U[j]}}{\sqrt{\sum_{u_\ell\in \U}{e^{\widehat{\ZM}_\U[i]\cdot \widehat{\ZM}_\U[\ell]}}}\cdot \sqrt{\sum_{u_\ell\in \U}{e^{\widehat{\ZM}_\U[j]\cdot \widehat{\ZM}_\U[\ell]}}}},
\end{equation}  
\end{small}
where $\widehat{\ZM}_\U$ is a normalized $\ZM_\U$ whose each $i$-th row satisfies
\begin{small}
\begin{equation}\label{eq:Z}
\widehat{\ZM}_\U[i] = {\ZM_\U[i]}/{\|\ZM_\U[i]\|}.
\end{equation}
\end{small}

MSA $s(u_i,u_j)$  is  symmetric, i.e., $s(u_i,u_j)=s(u_j,u_i)\ \forall{u_i,u_j\in \U}$. Additionally, by imposing a normalization, $-1\le \widehat{\ZM}_\U[i]\cdot \widehat{\ZM}_\U[j]\le 1$ $\forall{u_i,u_j\in \U}$, and hence, the MSA values w.r.t. any node $u_i\in \U$ are scaled to a similar range.


\stitle{Optimization Objective for $\ZM_\U$} Next, we focus on the obtainment of the feature vector $\widehat{\ZM}_\U[i]$ for each node $u_i\in \U$. A favorable choice might be {graph neural networks} (GNNs) \cite{kipf2016semi}, which, however, cannot be readily applied to ABGs as existing GNNs are primarily designed for general graphs, and it is rather costly to train classic GNNs.
As demystified by recent studies \cite{ma2021unified,yang2021attributes,zhu2021interpreting}, many popular GNNs models can be unified into an optimization framework from the perspective of numeric optimization, which essentially produces node feature vectors being smooth on nearby nodes in terms of the underlying graph.
Inspired by this finding, we extend this optimization framework to ABGs.
More specifically, its objective is as follows:
\begin{equation}\label{eq:Z-obj}
\min_{\ZM_\U}{(1-\alpha) \cdot \OO_a + \alpha\cdot \OO_g},
\end{equation}
which includes a non-negative coefficient $\alpha\in [0,1]$ and two terms: (i) a fitting term $\OO_{a}$ in Eq. \eqref{eq:Oa} aiming at ensuring $\ZM_\U$ is close to the input attribute vectors $\XM_\U$,
\begin{equation}\label{eq:Oa}
    \OO_a = \|\ZM_\U-\XM_\U\|_F^2
\end{equation}
and (ii) a regularization term $\OO_{g}$ in Eq. \eqref{eq:Og} constraining the feature vectors of two nodes with high connectivity to be similar.
\begin{small}
\begin{equation}\label{eq:Og}
\OO_g = \frac{1}{2}\sum_{u_i,u_j\in \U}{{\widehat{w}(u_i,u_j)}\cdot \left\|\frac{\ZM_\U[i]}{\sqrt{\DM_\U[i,i]}}-\frac{\ZM_\U[j]}{\sqrt{\DM_\U[j,j]}}\right\|^2}
\end{equation}
\end{small}

The regularization term requires scaling $\ZM_\U[i]$ of each node $u_i$ in Eq. \eqref{eq:Og} with a factor $1/\sqrt{\DM_\U[i,i]}$ to avoid distorting the values in $\ZM_\U[i]$ when $u_i$ connects to massive or scant links.
The weight $\widehat{w}(u_i,u_j)$ in Eq. \eqref{eq:Og} is defined by
\begin{small}
\begin{equation*}
\widehat{w}(u_i,u_j)=\sum_{v_\ell\in \N(u_i)\cap \N(u_j)}\frac{w(u_i,v_\ell)\cdot w(v_\ell,u_j)}{\DM_\V[\ell,\ell]},
\end{equation*}
\end{small}
which reflects the strength of connections between two homogeneous nodes $u_i$ and $u_j$ in $\U$ (e.g., researchers) via their common neighbors in the counterparty $\V$ (e.g., co-authored papers). As an example for illustration, consider researchers $u_3,u_4$ in Figure \ref{fig:ABGC}, $\widehat{w}(u_3,u_4)=\frac{1}{3}+\frac{1}{2}+\frac{1}{4}$, where the denominators $3$, $2$, and $4$ correspond to the numbers of authors in papers $v_3, v_4$ and $v_5$. Accordingly, 
$\widehat{w}(u_3,u_4)$ evaluates the overall contributions of $u_3,u_4$ to their collaborated research works in $\V$. Thus, the $\OO_g$ term in Eq. \eqref{eq:Z-obj} is to minimize the distance of feature vectors of researchers who have extensively collaborated with each other with high contributions.

The hyper-parameter $\alpha$ balances the attribute and topology information
encoded into $\ZM_\U$. In particular, when $\alpha=0$, feature vectors $\ZM_\U=\XM_\U$, and at the other extreme, i.e., $\alpha=1$, $\ZM_\U$ is entirely dependent on the topology of $\G$.

\stitle{Closed-form Solution of $\ZM_\U$}
Given an $\alpha$, Lemma \ref{lem:ZPX}\footnote{All proofs appear in Appendix \ref{sec:proofs}.} indicates that the optimal feature vectors $\ZM_\U$ to Eq. \eqref{eq:Z-obj} can be computed via iterative sparse matrix multiplications in Eq. \eqref{eq:PX} without undergoing expensive training.
\begin{lemma}\label{lem:ZPX}
When $\gamma\rightarrow\infty$, $\ZM_\U$ in Eq. \eqref{eq:PX} is the closed-form solution to the optimization problem in Eq. \eqref{eq:Z-obj}.
\begin{small}
\begin{equation}\label{eq:PX}
\ZM_\U = (1-\alpha)\sum_{r=0}^{\gamma}{\alpha^r \cdot  \left(\LM_\U\LM_\U^{\top}\right)^{r}}\XM_\U,
\end{equation}
where $\LM_\U$ is a normalized version of adjacency matrix $\BM_\U$, i.e.,
\begin{equation}\label{eq:L}
    \LM_\U = \DM_\U^{-\frac{1}{2}}\BM_\U\DM_\V^{-\frac{1}{2}}.
\end{equation}
\end{small}
\end{lemma}
In practice, we set $\gamma$ in Eq. \eqref{eq:PX} to a finite number (typically $5$) for efficiency. Intuitively, the computation of $\ZM_\U$ essentially aggregates attributes from other homogeneous nodes as per their multi-scale proximities (e.g., the strength of connections via multiple hops (at most $\gamma$ hops)) in $\G$.
As such, the feature vectors of nodes with numerous direct or indirect linkages will be more likely to be close, yielding a high MSA in Eq. \eqref{eq:s}.




\subsection{Objective Function}\label{sec:obj}
Based on the foregoing definitions of $k$-ABGC and MSA, we formulate the problem of $k$-ABGC as the following objective function:
\begin{small}
\begin{equation}\label{eq:obj1}
\min_{\C_1,\C_2,\cdots,\C_k} \sum_{\ell=1}^{k}\frac{1}{|\C_\ell|}\sum_{u_i\in \C_\ell, u_j\in \U\setminus\C_\ell}{s(u_i,u_j)},
\end{equation}
\end{small}
More precisely, Eq. \eqref{eq:obj1} is to identify $k$ disjoint clusters $\C_1,\C_2,\cdots,\C_k$ in $\U$ such that the average MSA of two nodes in different clusters is low. Meanwhile, with this optimization objective, the average MSA of any two nodes in the same cluster will be maximized; in other words, nodes within the same cluster are tight-knit.

According to \cite{shi2000normalized}, Eq. \eqref{eq:obj1} is an NP-complete combinatorial optimization problem.  Hence, the exact solution to Eq. \eqref{eq:obj1} is computationally infeasible when $\U$ contains a large number of nodes.
Moreover, the direct optimization of Eq. \eqref{eq:obj1} demands materializing $s(u_i,u_j)$ of every node pairs in $\U\times \U$. As such, deriving an approximate solution by optimizing Eq. \eqref{eq:obj1} with even a handful of epochs entails an $O(|\EDG|\cdot |\U|\cdot d_\U)$ computational cost and a quadratic space overhead $O(|\U|^2)$, rendering it incompetent for large ABGs.

\section{The \algo Algorithm}\label{sec:algo}
To address the above-said challenges, this section presents our {\em \underline{T}hree-\underline{P}hase \underline{O}ptimization} framework (\algo) to $k$-ABGC computation based on Eq. \eqref{eq:obj1} without explicitly constructing the MSA matrix.

\subsection{Synoptic Overview}
At a high level, \algo draws inspiration from the equivalence between the optimization objectives in Eq. \eqref{eq:obj1} and Eq. \eqref{eq:obj3}, as Lemma~\ref{lem:obj3}. 
\begin{lemma}\label{lem:obj3}
Eq. \eqref{eq:obj1} is equivalent to the following objective:
\begin{equation}\label{eq:obj3}
\min_{\YM\ge 0, \HM\ge 0}\|\RM-\YM\HM^{\top}\|^2_F\ \text{s.t. $\YM$ is an NCI matrix.}
\end{equation}
\end{lemma}
Specifically, if we can identify a matrix $\RM$ such that $\RM[i]\cdot \RM[j]= s(u_i,u_j)$ $\forall{u_i.u_j\in \U}$, the computation of $k$ non-overlapping clusters $\C_1$, $\C_2$, $\cdots$, $\C_k$ towards optimizing Eq. \eqref{eq:obj1} is equivalent to decomposing $\RM$ into two non-negative matrices $\YM$ and $\HM$, where $\YM$ represents a {\em normalized cluster indicator} (NCI) matrix $\YM\in \mathbb{R}^{|\U|\times k}$, as defined in Eq. \eqref{eq:Y}.
\begin{equation}\label{eq:Y}
\YM[i,\ell]=
\begin{cases}
\frac{1}{\sqrt{|\C_\ell|}}\quad \text{if $u_i$ belongs to in cluster $\C_\ell$},\\
0\quad \text{otherwise}.
\end{cases}
\end{equation}
According to Eq. \eqref{eq:Y}, for each node $u_i\in \U$, its corresponding vector $\YM[i]$ in the NCI matrix comprises solely one non-zero entry $\YM[i,\ell]$ indicating the clustering membership of $u_i$, and the value should be ${1}/{\sqrt{|\C_\ell|}}$. This characteristic ensures that $\YM$ is column-orthogonal, i.e., $\YM^{\top}\YM=\IM$. However, this constraint on $\YM$ renders the factorization of $\RM$ hard to converge. Instead of directly computing the exact $\YM$, we employ a two-step approximation strategy. More specifically, \algo first builds a $|\U|\times k$ matrix $\FY$ (a continuous version of $\YM$) which minimizes the factorization loss in Eq. \eqref{eq:obj4}:
\begin{equation}\label{eq:obj4}
\min_{\FY\ge 0, \HM\ge 0}\|\RM-\FY\HM^{\top}\|^2_F\ \text{s.t. $\FY^{\top}\FY=\IM$},
\end{equation}
in which the constraint on $\YM$ in Eq. \eqref{eq:obj3} is relaxed to be $\FY\ge 0$ and $\FY^{\top}\FY=\IM$. Afterward, the task is to transform $\FY$ into an NCI matrix $\YM$ by minimizing their difference about Eq. \eqref{eq:obj3}.


\begin{figure}[!t]
\centering
\includegraphics[width=0.95\columnwidth]{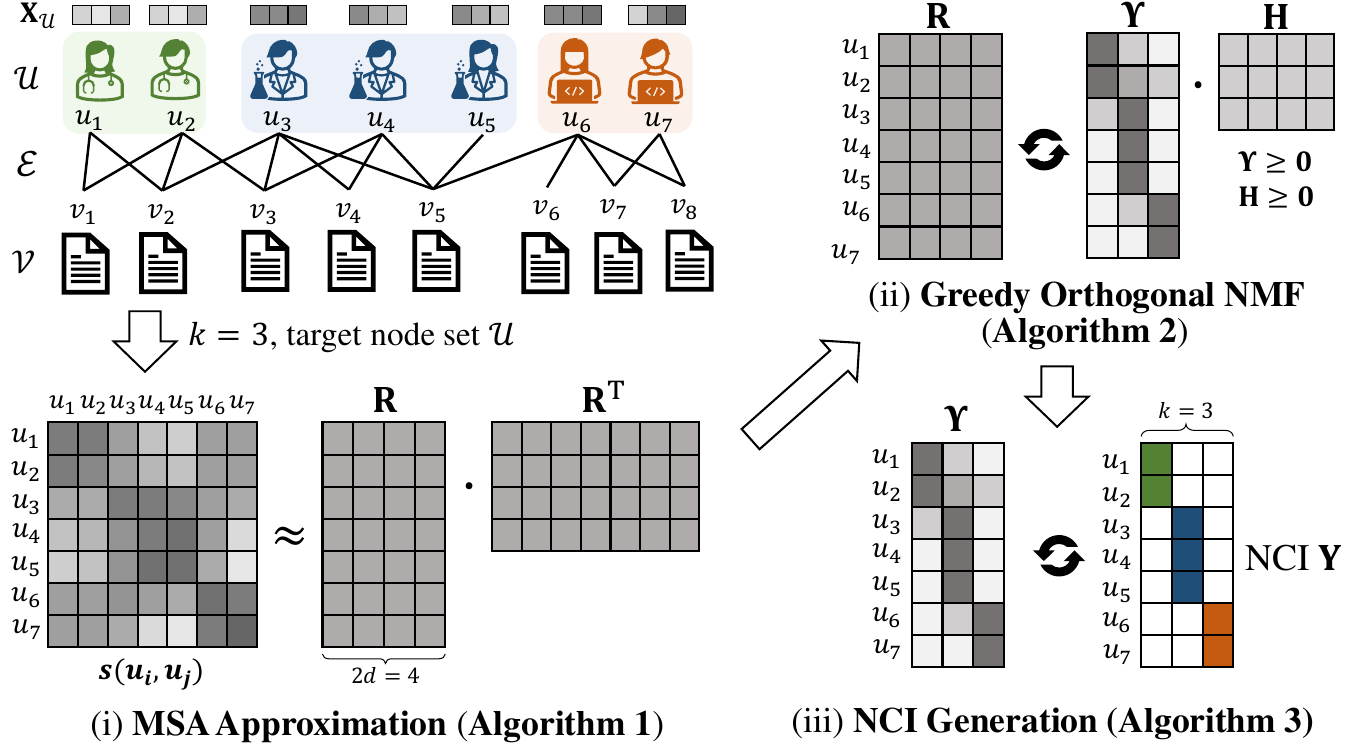}
\vspace{-3mm}
\caption{Overview of \algo}\label{fig:overview}
\end{figure}

As outlined in Figure \ref{fig:overview}, given an ABG $\G$, the number of $k$ of clusters, and the node set $\U$ to be partitioned as input, \algo outputs an approximate solution to the $k$-ABGC problem in Eq. \eqref{eq:obj1} through three phases: (i) constructing a low-dimensional matrix $\RM$ such that $\RM[i]\cdot \RM[j]\approx s(u_i,u_j)$ $\forall{u_i.u_j\in \U}$ without explicitly materializing the MSA of all node pairs (Algorithm~\ref{alg:aa}, Section \ref{subsec:MSAAprox}); (ii) factorizing $\RM$ as per Eq. \eqref{eq:obj4} to create a $\U\times k$ non-negative column-orthogonal matrix $\FY$ (Algorithm~\ref{alg:onmf}, Section \ref{sec:ONMF}); and (iii) effectively converting $\FY$ into an NCI $\YM$ (Algorithm~\ref{alg:round}, Section \ref{sec:NCI-gen}). In what follows, we elaborate on the algorithmic details of these three subroutines. Due to space limit, we defer 
the complexity analysis of them and \algo to Appendix~\ref{sec:complexity}.







\subsection{MSA Approximation via Random Features}\label{subsec:MSAAprox}

Algorithm~\ref{alg:aa} illustrates the pseudo-code of linearizing the approximate computation of MSA in Eq. \eqref{eq:s} as the matrix product $\RM\cdot\RM^{\top}$ with matrix $\RM$. The fundamental idea is to leverage and tweak the {\em random features} \cite{rahimi2007random,yu2016orthogonal} technique designed for approximating the Gaussian kernel $e^{-\|\mathbf{x}-\mathbf{y}\|^2/2}$ of any vectors $\mathbf{x}$ and $\mathbf{y}$. 

After taking as input the ABG $\G$ and parameters $\alpha,\gamma$, Algorithm~\ref{alg:aa} begins by calculating $\LM_\U$ according to Eq. \eqref{eq:L} and initializing $\ZM_\U$ as $\alpha \XM_\U$ (Lines 1-2). At Lines 3-4, we update $\ZM_\U$ via $\gamma$ iterations of the following matrix multiplication:
\begin{equation}\label{eq:update-Z}
\ZM_\U \gets \alpha\cdot(\XM_\U + \LM_\U\cdot (\LM_\U^{\top}\ZM_\U) ).
\end{equation}
Particularly, we structure the matrix multiplication $\LM_\U\LM_\U^{\top}\ZM_\U$ as $\LM_\U\cdot (\LM_\U^{\top}\ZM_\U)$ in Eq. \eqref{eq:update-Z} to boost the computation efficiency. Subsequently, Algorithm~\ref{alg:aa} transforms $\ZM_\U$ into $\widehat{\ZM}_\U$ by applying an $L_2$ normalization to each row in $\ZM_\U$ (Line 5) and then proceeds to constructing $\RM$ (Lines 6-9). 

To be specific, we first generate a $d_\U \times d_\U$ Gaussian random matrix $\GM$ with every entry sampled independently from the standard normal distribution (Line 6) and then apply a QR decomposition over it to get a $d_\U \times d_\U$ orthogonal matrix $\QM$ (Line 7). The matrix $\QM$ is distributed uniformly on the Stiefel manifold, i.e., the space of all orthogonal matrices \cite{muirhead2009aspects}. Next, Algorithm~\ref{alg:aa} calculates $\RM^{\prime}$ by
\begin{small}
\begin{equation}\label{eq:R-prime}
\RM^{\prime} = \sqrt{\frac{e}{d_\U}}\cdot \left(sin(\widehat{\ZM}^{\circ}_\U) \mathbin\Vert cos(\widehat{\ZM}^{\circ}_\U)\right)\in \mathbb{R}^{|\U| \times 2d_\U},
\end{equation}
\end{small}
where $\widehat{\ZM}^{\circ}_\U= \sqrt{d_\U}\cdot \widehat{\ZM}_\U\cdot {\QM}^{\top}$ and $\mathbin\Vert$ represents the horizontal concatenation operator for matrices (Line 8). Finally, in Line 9, the matrix $\RM$ is obtained by normalizing each row of $\RM^{\prime}$ as
\begin{small}
\begin{equation}\label{eq:Ri}
\RM[i] = \frac{\RM^{\prime}[i]}{\sqrt{\RM^{\prime}[i]\cdot \mathbf{r}}}\in \mathbb{R}^{2d_\U}\ \text{where $\mathbf{r} = \sum_{u_\ell\in \U}{\RM^{\prime}[\ell]}$}.
\end{equation}
\end{small}

\begin{algorithm}[!t]
\caption{MSA Approximation}\label{alg:aa}
\KwIn{An ABG $\G=(\U \cup \V, \EDG, \XM_\U)$, target node set $\U$, the balance coefficient $\alpha$, the number $\gamma$ of iterations}
\KwOut{Matrix $\RM$}
Calculate $\LM_\U$ according to Eq. \eqref{eq:L}\;
$\ZM_\U \gets \alpha\XM_\U$\;
\For{$r\gets 1$ to $\gamma$}{
Update $\ZM_\U$ according to Eq. \eqref{eq:update-Z}\;
}
Normalize $\ZM_\U$ as $\widehat{\ZM}_\U$ by Eq. \eqref{eq:Z}\;
Sample a Gaussian random matrix $\GM\in \mathbb{R}^{d_\U \times d_\U}$\;
Compute $\QM$ by a QR decomposition over $\GM$\;
Compute $\RM^{\prime}$ according to Eq. \eqref{eq:R-prime}\;
\lFor{$u_i\in \U$}{Compute $\RM[i]$ according to Eq. \eqref{eq:Ri}}
\Return{$\RM$}\;
\end{algorithm}

\begin{theorem}\label{lem:RR}
For any two nodes $u_i, u_j \in \U$, if $\RM$ is the output of Algorithm~\ref{alg:aa}, then the following inequality holds:
\begin{small}
\begin{equation*}
\frac{1 - \frac{17}{16d_\U^2} - \frac{1}{4d_\U}}{1 + \frac{17}{16d_\U^2} + \frac{1}{4d_\U}}\cdot s(u_i,u_j) \le \mathbb{E}[\RM[i]\cdot \RM[j]] \le \frac{1 + \frac{17}{16d_\U^2} + \frac{1}{4d_\U}}{1 - \frac{17}{16d_\U^2} - \frac{1}{4d_\U}}\cdot s(u_i,u_j)
\end{equation*}
\end{small}
\end{theorem}

Theorem \ref{lem:RR} indicates that $\mathbb{E}[\RM[i]\cdot \RM[j]]$ serves as an accurate estimator of $s(u_i,u_j)$, exhibiting extremely low bias, particularly because $d_\U$ often exceeds hundreds in practical scenarios.

\subsubsection{\bf SVD-based Attribute Dimension Reduction}\label{sec:d}
Although Algorithm~\ref{alg:aa} circumvents the need to construct the MSA for all node pairs, it remains tenaciously challenging when dealing with ABGs with vast attribute sets, i.e., $d_\U$ being large.
Recall that the major computation expenditure in Algorithm~\ref{alg:aa} lies at Lines 3-4 and Lines 7-8, which need $O(\gamma\cdot |\EDG|\cdot d_\U)$ and $O(d_\U^3+|\U|\cdot d_\U^2)$ time, respectively. As a result, when $d_\U$ is high, e.g., $d_\U = O(|\U|)$, the computational complexity of Algorithm~\ref{alg:aa} increases dramatically to be cubic, rendering it impractical for large-scale ABGs.

To address this, we refine the input attribute vectors $\XM_\U$ by reducing their dimension from $d_\U$ to a much smaller constant $d$ ($d\ll d_\U$). This approach aims to ensure that the $d$-dimensional approximation $\XM^{\prime}_\U$ of $\XM_\U$ still accurately preserves the MSA as per Eq. \eqref{eq:s}. This adjustment reduces the computational cost to a linear time complexity of $O(\gamma\cdot |\EDG| + |\U|)$ since $d$ is a constant.
To realize this idea, we first apply the top-$d$ {\em singular value decomposition} (SVD) over $\XM_\U$ to produce the decomposition result $\XM_\U\approx \GaM\SGVM\PsiM^{\top}$. Utilizing the column-orthogonal (semi-unitary) property of $\PsiM$, i.e., $\PsiM^{\top}\PsiM=\IM$, we have $\XM_\U\XM_\U^{\top}\approx \GaM\SGVM\PsiM^{\top}\PsiM\SGVM\GaM^{\top}=\GaM\SGVM^2\GaM^{\top}$,
%
implying 
\begin{equation}\label{eq:Xd}
\XM^{\prime}_\U=\GaM\SGVM\in \mathbb{R}^{|\U|\times d},
\end{equation}
which can be employed as a low-dimensional substitute of $\XM_\U$ input to Algorithm~\ref{alg:aa}.
Along this line, we can derive a low-dimensional version $\ZM_\U^{\prime}$ of feature vectors $\ZM_\U$ using the iterative process at Lines 2-4 in Algorithm~\ref{alg:aa} by simply replacing $\XM_\U$ as $\XM^{\prime}_\U$, i.e.,
\begin{small}
\begin{equation*}
\ZM_\U^{\prime}=(1-\alpha)\sum_{r=0}^{\infty}{\alpha^r \cdot  \left(\LM_\U\LM_\U^{\top}\right)^{r}}\XM^{\prime}_\U.
\end{equation*}
\end{small}

\begin{lemma}\label{lem:FF-ZZ}
Let $\GaM\SGVM\PsiM^{\top}$ be the exact top-$d$ SVD of $\XM_\U$.
\begin{small}
\begin{equation*}
\left|{\ZM_\U^{\prime}[i]\cdot \ZM_\U^{\prime}[j]}-{\ZM_\U[i]\cdot \ZM_\U[j]}\right| \le \frac{\sigma_{d+1}^2 \sqrt{\DM_\U[i,i]\cdot \DM_\U[j,j]}}{1-\alpha}
\end{equation*}
\end{small}
holds for every two nodes $u_i,u_j\in \U$, where $\sigma_{d+1}$ is the $(d+1)$-th largest singular value of $\XM_\U$.
\end{lemma}
Lemma \ref{lem:FF-ZZ} establishes the approximation guarantee of $\ZM_\U^{\prime}$, which theoretically assures the accurate approximation of the MSA defined in Eq. \eqref{eq:s}. Aside from the capabilities of preserving MSA and reducing computation load, this SVD-based trick can surprisingly denoise attribute data for enhanced clustering by its close connection to {\em principal component analysis} (PCA), as validated by our experiments in Section \ref{sec:exp-quality}.

\subsection{Greedy Orthogonal NMF}\label{sec:ONMF}




\begin{algorithm}[!t]
\caption{Greedy Orthogonal NMF}\label{alg:onmf}
\KwIn{Matrix $\RM$, the number $k$ of clusters, the number $T_f$ of iterations}
\KwOut{Matrix $\FY$}
$\GaM, \SGVM, \PsiM \gets \mathtt{RandomizedSVD}(\RM, k)$\;
Initialize $\FY$ and $\HM$ according to Eq. \eqref{eq:init-HY}\;
\For{$t\gets 1$ to $T_f$}{
    Update $\HM[j,\ell]\ \forall{1\le j\le 2d, 1\le \ell \le k}$ by Eq. \eqref{eq:update-Q}\;
    Update $\FY[i,\ell]\ \forall{u_i\in \U, 1\le \ell \le k}$ by Eq. \eqref{eq:update-Y}\;
}
\Return{$\FY$}\;
\end{algorithm}

Upon constructing $\RM\in \mathbb{R}^{|\U|\times 2d}$ (with $d=d_\U$ if the dimension reduction from Section \ref{sec:d} is not applied) in Algorithm~\ref{alg:aa}, \algo passes it to the second phase, i.e., conducting an orthogonal non-negative matrix factorization (NMF) of $\RM$ as in Eq. \eqref{eq:obj4} to create $\FY$. 
The pseudo-code of our solver to this problem is presented in Algorithm~\ref{alg:onmf}, iteratively updating $\FY$ and $\HM$ using an alternative framework towards optimizing the objective function in Eq.~\eqref{eq:obj4}. (Lines 3-5). 
Specifically, given the number $T_f$ of iterations and initial guess of $\HM$ and $\FY$, in each iteration, we first update each $(j,\ell)$-entry ($1\le j\le 2d$ and $1\le \ell \le k$) in $\HM$ following Eq. \eqref{eq:update-Q} while fixing $\FY$, and then update $\FY[i,\ell]$ for $u_i\in \U$ and $1\le \ell \le k$ as in Eq. \eqref{eq:update-Y} with $\HM$ fixed.
\begin{small}
\begin{equation}\label{eq:update-Q}
\HM[j,\ell]= \HM[j,\ell]\cdot \frac{(\RM^{\top}\FY)[j,\ell]}{(\HM\cdot (\FY^{\top}\FY))[j,\ell]}
\end{equation}
\end{small}
\begin{small}
\begin{equation}\label{eq:update-Y}
\FY[i,\ell]= \FY[i,\ell]\cdot \sqrt{\frac{(\RM\HM)[i,\ell]}{(\FY\cdot (\FY^{\top}\cdot(\RM\HM)))[i,\ell]}}
\end{equation}
\end{small}

The above update rules for solving Eq. \eqref{eq:obj4} can be derived by utilizing the {\em auxiliary function approach} \cite{lee2000algorithms} with Lagrangian multipliers in convex optimization, whose convergence is guaranteed by the monotonicity theorem \cite{ding2006orthogonal}.
Note that we reorder the matrix multiplications $\HM\FY^{\top}\FY$ and $\FY\FY^{\top}\RM\HM$ in Eq. \eqref{eq:update-Q} and \eqref{eq:update-Y} to $\HM\cdot (\FY^{\top}\FY)$ and $\FY\cdot (\FY^{\top}\cdot(\RM\HM))$, respectively, so as to avert materializing $2d \times |\U|$ dense matrix $\HM\FY^{\top}$ and $|\U|\times |\U|$ dense matrix $\FY\FY^{\top}$. As such, the computational complexities of updating $\HM$ and $\FY$ in Eq. \eqref{eq:update-Q} and \eqref{eq:update-Y} are reduced to $O(|\U|dk+|\U|k^2)$ per iteration.


The aforementioned computation is still rather costly due to the numerous iterations needed for the convergence of $\FY$ and $\HM$, especially when $\FY$ and $\HM$ are initialized randomly.
We resort to a greedy seeding strategy to expedite convergence, as in many optimization problems. That is, we carefully select a good initialization of $\FY$ and $\HM$ in a fast but theoretically grounded manner. As described in Lines 1-2 in Algorithm~\ref{alg:aa}, 
we set $\FY$ and $\HM$ as follows:
\begin{equation}\label{eq:init-HY}
\FY = \GaM,\ \HM = \PsiM\SGVM,
\end{equation}
where $\GaM$ and $\PsiM$ are the top-$k$ left and right singular vectors of $\RM$, respectively, and $\SGVM$ is a diagonal matrix whose diagonal entries are top-$k$ singular values of $\RM$, which are obtained by invoking the {\em truncated} randomized SVD algorithm \cite{halko2011finding} with $\RM$ and $k$. Note that this routine consumes $O(|\U|dk+(\U+d)k^2)$ time \cite{halko2011finding} and can be done efficiently in practice in virtue of its randomized algorithmic design as well as the highly-optimized libraries (LAPACK and BLAS) for matrix operations under the hood.

Given the fact that singular vectors $\FY=\GaM$ are column-orthogonal, i.e., $\FY^{\top}\FY=\IM$, the Eckart-Young Theorem \cite{gloub1996matrix} (Theorem \ref{lem:eym} in Appendix \ref{sec:proofs}) pinpoints that Eq. \eqref{eq:init-HY} offers the optimal solution to Eq. \eqref{eq:obj4} when the non-negative constraints over $\FY$ and $\HM$ are relaxed. In simpler terms, Eq. \eqref{eq:init-HY} immediately gains a rough solution to our optimization objective in Eq. \eqref{eq:obj4}, thereby drastically curtailing the number of iterations needed for Lines 3-5.


\subsection{Effective NCI Generation}\label{sec:NCI-gen}


\begin{algorithm}[!t]
\caption{Effective NCI Generation}\label{alg:round}
\KwIn{Matrix $\FY$ and the number $T_g$ of iterations}
\KwOut{The NCI matrix $\YM$}
$\TM=\IM$\;
\For{$t\gets 1$ to $T_g$}{
    \For{$u_i\in \mathcal{U}$}{
        Determine $\ell^\ast$ by Eq. \eqref{eq:ell-ast}\;
        Update $\YM$ by Eq. \eqref{eq:Y-update}\;
    }
    Normalize $\YM$ such that each column has a unit $L_2$ norm\;
    $\TM \gets \FY^{\top}\YM$\;
}
\Return{$\YM$}\;
\end{algorithm}

In its final stage, \algo generates an NCI matrix $\YM$ by minimizing the ``difference'' between $\FY$ returned by Algorithm~\ref{alg:onmf} 
and the target NCI matrix $\YM$. Recall from Eq. \eqref{eq:obj3}, our original objective is to find a $|\U|\times k$ NCI matrix $\YM$ and a $2d\times k$ non-negative $\HM$ such that the total squared reconstruction error $\|\RM-\YM\HM^{\top}\|^2_F=\sum_{u_i\in \U}{\sum_{j=1}^{d}{(\RM[i,j]-\YM[i]\cdot \HM[j])^2}}$ is minimized. Considering $\FY$ is a continuous version of $\YM$ (relaxing the constraint in Eq. \eqref{eq:Y}), $\|\RM-\FY\HM^{\top}\|^2_F$ is capable of attaining a strictly lower reconstruction error compared to $\|\RM-\YM\HM^{\top}\|^2_F$. Therefore, an ideal solution $\YM$ to Eq. \eqref{eq:obj3} ensures that $\|\RM-\YM\HM^{\top}\|^2_F$ closely approximates $\|\RM-\FY\HM^{\top}\|^2_F$ in Eq. \eqref{eq:obj4}. Mathematically, the conversion from matrix $\FY$ into the NCI matrix $\YM$ can be formulated as the minimization of the difference of their reconstruction errors, i.e., $\left|\|\RM-\YM\HM^{\top}\|^2_F-\|\RM-\FY\HM^{\top}\|^2_F \right|=\left|\Tr((\YM\YM^{\top}-\FY\FY^{\top})\cdot \RM\RM^{\top}) \right|$ by Lemma \ref{lem:RYH-RGammaH}.
\begin{lemma}\label{lem:RYH-RGammaH} The following equation holds:
\begin{small}
\begin{equation}\label{eq:obj5}
\left|\|\RM-\YM\HM^{\top}\|^2_F-\|\RM-\FY\HM^{\top}\|^2_F \right|=\left|\Tr((\YM\YM^{\top}-\FY\FY^{\top})\cdot \RM\RM^{\top}) \right|.
\end{equation}
\end{small}
\end{lemma}

Further, we reformulate the problem as follows:
\begin{equation}\label{eq:obj6}
\min_{\TM, \YM}{\|\FY\TM-\YM\|_2}\ \text{s.t.}\ \text{$\TM\TM^{\top}=\IM$ and $\YM$ is an NCI matrix},
\end{equation}
which implies that, if the NCI matrix $\YM$ and the $k\times k$ row-orthogonal matrix $\TM$ minimize $\|\FY\TM-\YM\|_2$, $\YM\YM^{\top}-\FY\FY^{\top}\approx \FY\TM\TM^{\top}\FY^{\top}-\FY\FY^{\top}\approx 0$ holds and the objective loss in Eq. \eqref{eq:obj5} is therefore minimized. 

To solve Eq. \eqref{eq:obj6}, we develop Algorithm~\ref{alg:round} in \algo, which obtains the NCI matrix $\YM$ through an iterative framework wherein $\TM$ and $\YM$ are refined in an alternative fashion till convergence. Initially, Algorithm~\ref{alg:round} starts by taking as input the matrix $\FY$ and the number $T_g$ of iterations and initializing $\TM$ as a $k\times k$ identity matrix (Line~1). 
It then launches an iterative process at Lines 2-7 
to jointly refine $\YM$ and $\TM$. Specifically, in each of the $T_g$ iterations, \algo first determines the cluster id of each node $u_i\in \U$ via (Line 4)
\begin{equation}\label{eq:ell-ast}
\ell^\ast= \arg\max_{1\le \ell\le k}{\FY[i]\cdot \TM[:,\ell]}
\end{equation}
and then updates the cluster indicator $\YM[i]$ of node $u_i$ as follows (Line 5):
\begin{equation}\label{eq:Y-update}
\YM[i,\ell] = \begin{cases}
1 &\ \text{if $\ell=\ell^\ast$,} \\
0 &\ \text{otherwise},
\end{cases}
\ \forall{1\le \ell \le k}.
\end{equation}
Each column in $\YM$ is later $L_2$-normalized, i.e., 
\begin{equation}\label{eq:y-norm}
\textstyle \ \forall{1\le \ell\le k}\ \sqrt{\sum_{u_i\in \U}\YM[i,\ell]^2}=1,
\end{equation}
in accordance with the NCI constraint in Eq. \eqref{eq:Y} (Line 6). 
In a nutshell, Lines 3-6 optimizes Eq. \eqref{eq:obj6} by updating $\YM$ with $\TM$ fixed. To explain, recall the constraint of the NCI matrix $\YM$ stated in Eq. \eqref{eq:Y}, each row of $\YM$ has merely one non-zero entry. Hence, by locating the column id $\ell^\ast$ whose corresponding entry $(\FY\TM)[i,\ell^\ast]$ is maximum in the $i$-th row of $\FY\TM$ (i.e., Eq. \eqref{eq:ell-ast}) and meanwhile updating $\YM[i]$ as Eqs.~\eqref{eq:Y-update} and ~\eqref{eq:y-norm} as Lines 5-6, 
the distance between $\FY\TM$ and $\YM$ in Eq. \eqref{eq:obj6} is naturally minimized.

\begin{table}[!t]
\centering
\renewcommand{\arraystretch}{1.0}
\caption{Attributed Bipartite Graphs}\label{tbl:datasets}
\vspace{-2mm}
\resizebox{\columnwidth}{!}{%
\begin{tabular}{l|c|c|c|c|c}
\hline
{\bf Name} & {{\em CiteSeer}}  & {{\em Cora}}  & {{\em MovieLens}}  & {{\em Google}}  & {{\em Amazon}}  \\ \hline
$|\U|$ & 1,237 & 1,312 & 6,040 & 64,527 & 2,330,066 \\ 
$|\V|$ & 742 & 789 & 3,883 & 868,937 & 8,026,324 \\ 
$|\EDG|$ & 1,665 & 2,314 & 1,000,209 & 1,487,747 & 22,507,155 \\ 
$d_\U$ & 3,703 & 1,433 & 30 & 1,024 & 800 \\ 
$d_\V$ & 3,703 & 1,433 & 21 & - & - \\ 
$k$ & 6 & 7 & 21 & 5 & 3 \\ 
\hline
 \end{tabular}
 }
\end{table}

\begin{table*}[!t]
\centering
\renewcommand{\arraystretch}{0.8}
\caption{Clustering Quality (Larger ACC, NMI, and ARI indicate higher clustering quality).}\vspace{-3mm}
\begin{small}
\addtolength{\tabcolsep}{-0.25em}
\begin{tabular}{c|ccc | ccc | ccc | ccc | ccc | ccc | ccc | c}
\hline
\multirow{2}{*}{\bf Method} & \multicolumn{3}{c|}{\bf{ {\em CiteSeer}}} & \multicolumn{3}{c|}{\bf{ {\em Cora}}}  & \multicolumn{3}{c|}{\bf{ {\em MovieLens}}}  & \multicolumn{3}{c|}{\bf{ {\em Google}}} & \multicolumn{3}{c|}{\bf{ {\em Amazon}}} & \multirow{2}{*}{\bf Rank}  \\ \cline{2-16}
& ACC & NMI & ARI & ACC& NMI & ARI & ACC & NMI & ARI & ACC & NMI & ARI & ACC & NMI & ARI  & \\ 
\hline
\textsf{KMeans} \cite{hartigan1979algorithm} 	&	\cellcolor{gray!15} 0.526	&\cellcolor{gray!30}	0.277	&	0.225	&	0.431	&	0.229	&	0.137	&	0.298	&	0.363	&	0.170	&	0.370	&	0.053	&	0.012	&\cellcolor{gray!30}	0.502	&	0.038	&	0.079	&	5.4	     \\ 
\textsf{SpecClust} \cite{von2007tutorial} 	&	0.222	&	0.017	&	-0.001	&	0.311	&	0.026	&	0.003	&	0.318	&	0.392	&	0.197	&	- &	- &	- &	- &	- &	- &	11.27	    \\
\textsf{NMF} \cite{xu2003document} 	&	0.508	&	0.222	&\cellcolor{gray!15}	0.228	&	0.380	&	0.165	&	0.110	&	\cellcolor{gray!15}0.568	&	\cellcolor{gray!15}0.611	&	\cellcolor{gray!15}0.482	&	0.384	&	0.069	&	0.062	&	0.390	&	0.006	&	0.015	&	5.6	    \\ \hline
\textsf{SCC} \cite{dhillon2001co} 	&	0.243	&	0.015	&	0.012	&	0.280	&	0.040	&	0.017	&	0.128	&	0.030	&	0.004	&	0.425	&	0.038	&	0.038	&	0.470	&	0.018	&	-0.016	&	10.27	     \\ 
\textsf{SBC} \cite{kluger2003spectral} 	&	0.239	&	0.015	&	0.012	&	0.262	&	0.059	&	0.023	&	0.116	&	0.035	&	0.006	&	0.394	&	0.006	&	-0.003	&	0.485	&	0.003	&	0.005	&	10.87	    \\
\textsf{CCMOD} \cite{ailem2015co} 	&	0.200	&	0.010	&	0.003	&	0.264	&	0.066	&	0.040	&	0.141	&	0.029	&	0.010	&	OOM	&	OOM	&	OOM	&	OOM	&	OOM	&	OOM	&	12.27	     \\
\textsf{SpecMOD} \cite{labiod2011co} 	&	0.238	&	0.012	&	0.003	&	0.290	&	0.023	&	-0.007	&	0.125	&	0.033	&	0.009	&	OOM	&	OOM	&	OOM	&	OOM	&	OOM	&	OOM	&	12.6	     \\
\textsf{InfoCC} \cite{dhillon2003information} 	&	0.235	&	0.013	&	0.007	&	0.223	&	0.035	&	0.018	&	0.095	&	0.036	&	0.007	&	0.277	&	0.008	&	0.008	&	0.378	&	0.007	&	0.003	&	11.54	     \\ 
\textsf{DeepCC} \cite{xu2019deep} 	&	0.205	&	0.013	&	0.004	&	0.213	&	0.014	&	0.006	&	0.093	&	0.027	&	0.004	&	0.324	&	0.105	&	0.031	&	- &	- &	- &	12.67	      \\
\textsf{HOPE}~\cite{YangShi23}   	&	0.473	&	0.169	&	0.288	&	0.268	&	0.025	&	0.043	&	0.115	&	0.009	&	0.037	&	0.269	&	0	&	0	&	0.452 &	0 &	0.002 &	 \\ \hline
\textsf{ACMin} \cite{yang2021effective} 	&	0.450	&	0.143	&	0.141	&\cellcolor{gray!15}	0.650	&	0.470	&\cellcolor{gray!30}	0.410	&	0.122	&	0.032	&	0.009	&	0.312	&	0.023	&	0.020	&	0.428	&	0.012	&	0.003	&	7.67	      \\
\textsf{GRACE} \cite{fanseu2023grace} 	&	0.469	&	0.209	&	0.199	&	0.351	&	0.136	&	0.103	&	0.298	&	0.249	&	0.195	&	0.420	&	0 &	0 &	0.427	&	0.008	&	0 &	7.47	      \\ 
\textsf{AGCC} \cite{zhang2019attributed} 	&	0.448	&	0.144	&	0.153	&\cellcolor{gray!15}	0.650	&\cellcolor{gray!30}	0.517	&	0.406	&	0.538	&	0.589	&	0.480	&	OOM	&	OOM	&	OOM	&	OOM	&	OOM	&	OOM	&	7.34	    \\
\textsf{Dink-Net} \cite{Dink-Net}  	&	0.308	&	0	&	0	&	0.231	&	0.004	&	0.007	&	0.123	&	0.001	&	0	&\cellcolor{gray!30}	0.429	&	0	&	0	&	OOM	&	OOM	&	OOM	&	 	    \\ \hline
\textsf{node2vec} \cite{grover2016node2vec} 	&	0.209	&	0.007	&	0.001	&	0.220	&	0.008	&	0 &	0.074	&	0.011	&	0 &	0.280	&	0 &	0 &	- &	- &	- &	14.8	     \\
\textsf{BiNE} \cite{gao2018bine} 	&	0.196	&	0.005	&	0 &	0.174	&	0.005	&	-0.002	&	0.071	&	0.012	&	0 &	- &	- &	- &	- &	- &	- &	15.54	     \\
\textsf{GEBE} \cite{yang2022scalable} 	&	0.231	&	0.013	&	0.002	&	0.293	&	0.014	&	0.006	&	0.095	&	0.014	&	-0.002	&\cellcolor{gray!15}	0.428	&	0 &	0 &	0.489	&	0 &	0 &	12.14	     \\
\textsf{PANE} \cite{yang2020scaling,yang2023pane} 	&	0.443	&	0.154	&	0.136	&	0.537	&\cellcolor{gray!55}	0.526	&	0.339	&\cellcolor{gray!30}	0.855	&\cellcolor{gray!30}	0.923	&	\cellcolor{gray!30} 0.838	&	0.359	&\cellcolor{gray!30}	0.127	&\cellcolor{gray!15}	0.070	&	\cellcolor{gray!15} 0.497	&\cellcolor{gray!55}	0.083	&\cellcolor{gray!30}	0.102	&\cellcolor{gray!15}	4.2	     \\
\textsf{BiANE} \cite{huang2020biane} 	&	0.259	&	0.057	&	0.016	&	0.341	&	0.239	&	0.080	&	0.091	&	0.053	&	0.013	&	- &	- &	- &	- &	- &	- &	10.21	     \\ \hline
\algo ($d=d_\U$) 	&	\cellcolor{gray!30} 0.541	&\cellcolor{gray!15}	0.256	&\cellcolor{gray!30}	0.265	&\cellcolor{gray!30}	0.662	&\cellcolor{gray!15}	0.477	&\cellcolor{gray!15}	0.408	&\cellcolor{gray!55}	0.931	&\cellcolor{gray!55}	0.961	&\cellcolor{gray!55}	0.957	&	0.367	&\cellcolor{gray!15}	0.112	&\cellcolor{gray!30}	0.091	&\cellcolor{gray!30}	0.502	&\cellcolor{gray!15}	0.045	&\cellcolor{gray!15}	0.091	&\cellcolor{gray!30}	2.54	     \\
\algo 	& \cellcolor{gray!55}	0.625	&\cellcolor{gray!55}	0.322	&\cellcolor{gray!55}	0.348	&\cellcolor{gray!55}	0.671	&	0.475	&\cellcolor{gray!55}	0.416	&\cellcolor{gray!55}	0.931	&\cellcolor{gray!55}	0.961	&\cellcolor{gray!55}	0.957	&\cellcolor{gray!55}	0.444	&\cellcolor{gray!55}	0.135	&\cellcolor{gray!55}	0.138	&\cellcolor{gray!55}	0.504	&\cellcolor{gray!30}	0.055	&\cellcolor{gray!55}	0.104	&\cellcolor{gray!55}	1.27	     \\ \hline
\end{tabular}
\end{small}
\label{tbl:node-clustering}
\vspace{0ex}
\end{table*}

With the refined $\YM$ at hand, the subsequent work turns into updating the $k\times k$ matrix $\TM$ towards optimizing
\begin{equation*}
\min_{\TM}{\|\FY\TM-\YM\|_2}\ \text{s.t.}\ \text{$\TM\TM^{\top}=\IM$}.
\end{equation*}
Given $\YM$, the minimizer to this problem is $\TM=\FY^{\top}\YM$ by utilizing Lemma 4.14 in \cite{woodruff2014sketching}. Therefore, $\TM$ is updated to $\FY^{\top}\YM$ at Line 7.

After repeating the above procedure for $T_g$ iterations, \algo returns $\YM$ as the final clustering result. Practically, a dozen iterations are sufficient to yield high-caliber $\YM$, as validated in Section~\ref{sec:exp-param}.




\section{Experiments}\label{sec:exp}
In this section, we experimentally evaluate our proposed $k$-ABGC method \algo against 19 competitors over five real ABGs in terms of clustering quality and efficiency.
All the experiments are conducted on a Linux machine powered by 2 Xeon Gold 6330 @2.0GHz CPUs and 1TB RAM.
For reproducibility, the source code and datasets are available at \url{https://github.com/HKBU-LAGAS/TPC}.

\subsection{Experimental Setup}\label{sec:exp-setup}

\stitle{Datasets} 
Table \ref{tbl:datasets} lists the statistics of the five datasets used in the experimental study.
$|\U|$, $|\V|$, and $|\EDG|$ denote the cardinality of two disjoint node sets $\U$, $\V$, and edge set $\EDG$ of $\G$, respectively, while $d_{\U}$ (resp. $d_{\V}$) stands for the dimensions of attribute vectors of nodes in $\U$ (resp. $\V$). The number of ground-truth clusters of nodes $\U$ in $\G$ is $k$.
{\em Citeseer} and {\em Cora} are synthesized from real citation graphs in \cite{kipf2016semi} by dividing nodes in each cluster into two equal-sized partitions (i.e., $\U$ and $\V$) and removing intra-partition edges and isolated nodes as in \cite{xie2022bgnn}. In particular, nodes represent publications, edges denote their citation relationships, and labels correspond to the fields of study.
The well-known {\em MovieLens} dataset \cite{harper2015movielens} comprises user-movie ratings, where clustering labels are users' occupations in $\U$. {\em Google} and {\em Amazon} are extracted from the Google Maps \cite{yan2023personalized} and Amazon review dataset \cite{he2016ups}, where edges represent the reviews on restaurants and books posted by users.

\stitle{Competitors and Parameters}
We compare \algo against 19 existing methods, which can be categorized into four groups as follows:
\begin{itemize}[leftmargin=*]
\item {\em Data Clustering}: \textsf{KMeans} \cite{hartigan1979algorithm}, \textsf{NMF} \cite{xu2003document}, and \textsf{SpecClust} \cite{von2007tutorial};
\item {\em Network Embedding}: \textsf{node2vec} \cite{grover2016node2vec}, \textsf{BiNE} \cite{gao2018bine}, \textsf{BiANE} \cite{huang2020biane}, \textsf{PANE} \cite{yang2020scaling,yang2023pane}, and \textsf{GEBE} \cite{yang2022scalable};
\item {\em Attributed Graph Clustering}: \textsf{AGCC} \cite{zhang2019attributed}, \textsf{ACMin} \cite{yang2021effective}, \textsf{GRACE}~\cite{fanseu2023grace}, \textsf{Dink-Net} \cite{Dink-Net};
\item {\em Bipartite Graph Clustering}: \textsf{SCC} \cite{dhillon2001co}, \textsf{SBC} \cite{kluger2003spectral}, \textsf{InfoCC} \cite{dhillon2003information}, \textsf{SpecMOD} \cite{labiod2011co}, \textsf{CCMOD} \cite{ailem2015co}, \textsf{DeepCC} \cite{xu2019deep}, and \textsf{HOPE}~\cite{YangShi23}.
\end{itemize}



Unless otherwise specified, on all datasets, we set the numbers $T_f$ and $T_g$ of iterations required by Algorithms \ref{alg:onmf} and \ref{alg:round} in our proposed \algo to $5$ and $20$, respectively. Regarding parameters $\alpha$ and $\gamma$, by default, we set $\alpha=0.6, \gamma=6$ on {\em CiteSeer} and {\em MovieLens}, $\alpha=0.9, \gamma=10$ on {\em Cora} and {\em Google}, and $\alpha=0.5, \gamma=1$ on {\em Amazon}, respectively. To deal with the high attribute dimensions $d_\U$ of the {\em CiteSeer}, {\em Cora}, {\em Google}, and {\em Amazon} datasets, we set their new attribute dimensions $d$ in Section \ref{sec:d} to $32$, $128$, $32$, and $64$, respectively. We refer to the version of \algo~{\em without} the attribute dimension reduction module in Section \ref{sec:d} as \algo ($d=d_\U$).
More implementation details of our method and baselines are  in Appendix \ref{sec:eval_metrics}.



\pgfplotsset{ every non-boxed x axis/.append style={x axis line style=-} }
\pgfplotsset{ every non-boxed y axis/.append style={y axis line style=-} }
\begin{figure*}[!t]
\centering
\begin{small}
\subfloat[\em{Cora}]{
\begin{tikzpicture}[scale=1]
\begin{axis}[
    height=\columnwidth/2.3,
    width=\columnwidth/2.3,
    xtick=\empty,
    ybar=1.5pt,
    bar width=0.22cm,
    enlarge x limits=true,
    ylabel={\em running time} (sec),
    xticklabel=\empty,
    ymin=0.01,
    ymax=25,
    ytick={0.01,0.1,1,10},
    ymode=log,
    log origin y=infty,
    log basis y={10},
    every axis y label/.style={at={(current axis.north west)},right=10mm,above=0mm},
    legend style={draw=none, at={(1.02,1.02)},anchor=north west,cells={anchor=west},font=\scriptsize},
    legend image code/.code={ \draw [#1] (0cm,-0.1cm) rectangle (0.3cm,0.15cm); },
    ]

\addplot [pattern={grid}] coordinates {(1,0.47) };
\addplot [pattern={crosshatch dots}] coordinates {(1,4.47) }; 
\addplot [pattern=north west lines] coordinates {(1,18.9994) }; 
\addplot [pattern=crosshatch] coordinates {(1,0.877064) };
\addplot [pattern=north east lines] coordinates {(1,2.469) }; 
\addplot [pattern=horizontal lines] coordinates {(1,0.512) }; 
\addplot [pattern=checkerboard] coordinates {(1,0.07) }; 


\legend{\algo,\algo ($d=d_\U$),\textsf{AGCC},\textsf{ACMin},\textsf{PANE},\textsf{KMeans},\textsf{NMF}}
\end{axis}
\end{tikzpicture}\hspace{-3mm}\label{fig:time-cora}%
}%
\subfloat[{\em MovieLens}]{
\begin{tikzpicture}[scale=1]
\begin{axis}[
    height=\columnwidth/2.3,
    width=\columnwidth/2.3,
    xtick=\empty,
    ybar=1.5pt,
    bar width=0.22cm,
    enlarge x limits=true,
    ylabel={\em running time} (sec),
    xticklabel=\empty,
    ymin=0.1,
    ymax=120,
    ytick={0.1,1,10,100},
    ymode=log,
    log origin y=infty,
    log basis y={10},
    every axis y label/.style={at={(current axis.north west)},right=10mm,above=0mm},
    legend style={draw=none, at={(1.02,1.02)},anchor=north west,cells={anchor=west},font=\scriptsize},
    legend image code/.code={ \draw [#1] (0cm,-0.1cm) rectangle (0.3cm,0.15cm); },
    ]

\addplot [pattern={grid}] coordinates {(1,1.36) }; 
\addplot [pattern={crosshatch dots}] coordinates {(1,1.36) }; 
\addplot [pattern=north east lines] coordinates {(1,4.70) }; 
\addplot [pattern=checkerboard] coordinates {(1,0.466) }; 
\addplot [pattern=north west lines] coordinates {(1,86.53) }; 
\addplot [pattern=horizontal lines] coordinates {(1,5.986) }; 
\addplot [pattern=crosshatch] coordinates {(1,2.313) };

\legend{\algo,\algo ($d=d_\U$),\textsf{PANE},\textsf{NMF},\textsf{AGCC},\textsf{SpecClust},\textsf{GRACE}}

\end{axis}
\end{tikzpicture}\hspace{-2mm}\label{fig:time-movielens}%
}%
\subfloat[{\em Google}]{
\begin{tikzpicture}[scale=1]
\begin{axis}[
    height=\columnwidth/2.3,
    width=\columnwidth/2.3,
    xtick=\empty,
    ybar=1.5pt,
    bar width=0.22cm,
    enlarge x limits=true,
    ylabel={\em running time} (sec),
    xticklabel=\empty,
    ymin=1,
    ymax=120000,
    ytick={1,10,100,1000,10000,100000},
    ymode=log,
    log origin y=infty,
    log basis y={10},
    every axis y label/.style={at={(current axis.north west)},right=10mm,above=0mm},
    legend style={draw=none, at={(1.02,1.02)},anchor=north west,cells={anchor=west},font=\scriptsize},
    legend image code/.code={ \draw [#1] (0cm,-0.1cm) rectangle (0.3cm,0.15cm); },
    ]

\addplot [pattern={grid}] coordinates {(1,28.7) }; 
\addplot [pattern={crosshatch dots}] coordinates {(1,571.5) }; 
\addplot [pattern=north east lines] coordinates {(1,1379.5) }; 
\addplot [pattern=checkerboard] coordinates {(1,7.934) };  
\addplot [pattern=north west lines] coordinates {(1,8.217) }; 
\addplot [pattern=crosshatch] coordinates {(1,89628) };
\addplot [pattern=horizontal lines] coordinates {(1,19.169) };

\legend{\algo,\algo ($d=d_\U$),\textsf{PANE},\textsf{NMF},\textsf{SCC},\textsf{DeepCC},\textsf{KMeans}}

\end{axis}
\end{tikzpicture}\hspace{-2mm}\label{fig:time-google}%
}%
\subfloat[{\em Amazon}]{
\begin{tikzpicture}[scale=1]
\begin{axis}[
    height=\columnwidth/2.3,
    width=\columnwidth/2.3,
    xtick=\empty,
    ybar=1.5pt,
    bar width=0.22cm,
    enlarge x limits=true,
    ylabel={\em running time} (sec),
    xticklabel=\empty,
    ymin=1,
    ymax=20000,
    ymode=log,
    ytick={1,10,100,1000,10000},
    log basis y={10},
    every axis y label/.style={at={(current axis.north west)},right=10mm,above=0mm},
    legend style={draw=none, at={(1.02,1.02)},anchor=north west,cells={anchor=west},font=\scriptsize},
    legend image code/.code={ \draw [#1] (0cm,-0.1cm) rectangle (0.3cm,0.15cm); },
    ]

\addplot [pattern={grid}] coordinates {(1,178) }; 
\addplot [pattern={crosshatch dots}] coordinates {(1,1161.65) }; 
\addplot [pattern=north east lines] coordinates {(1,14802.35) };
\addplot [pattern=horizontal lines] coordinates {(1,24.74) }; 
\addplot [pattern=crosshatch] coordinates {(1,935.886) }; 
\addplot [pattern=north west lines] coordinates {(1,221.96) };  
\addplot [pattern=checkerboard] coordinates {(1,18.77) };

\legend{\algo,\algo ($d=d_\U$),\textsf{PANE},\textsf{KMeans},\textsf{ACMin},\textsf{SBC},\textsf{NMF}}

\end{axis}
\end{tikzpicture}\hspace{0mm}\label{fig:time-amazon}%
}%
\vspace{-4mm}
\end{small}
\caption{Running time in seconds.} \label{fig:time}
\end{figure*}
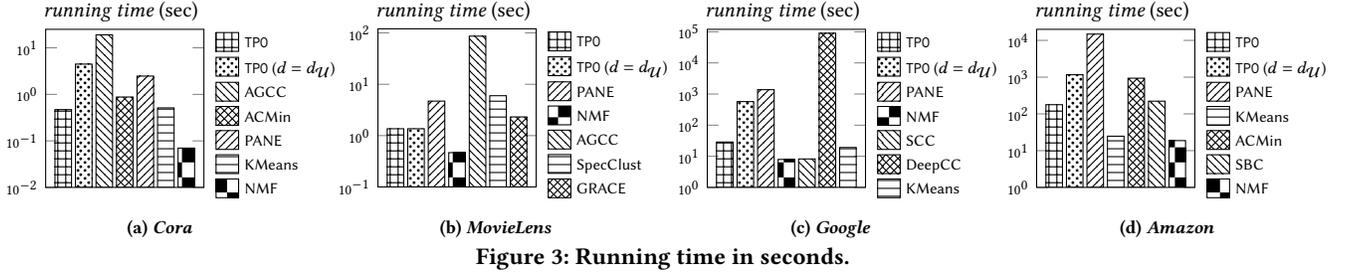

 \begin{figure*}[!t]
\centering
\begin{small}
\begin{tikzpicture}
    \begin{customlegend}[legend columns=4,
        legend entries={CiteSeer, Cora, MovieLens, Google},
        legend style={at={(0.45,1.35)},anchor=north,draw=none,font=\footnotesize,column sep=0.2cm}]
    \addlegendimage{line width=0.2mm,mark size=4pt,mark=diamond}
    \addlegendimage{line width=0.2mm,mark size=4pt,mark=triangle}
    \addlegendimage{line width=0.2mm,mark size=4pt,mark=o}
    \addlegendimage{line width=0.2mm,mark size=4pt,mark=square}
    \end{customlegend}
\end{tikzpicture}
\\[-\lineskip]
\vspace{-4mm}
\subfloat[Varying $\alpha$]{
\begin{tikzpicture}[scale=1,every mark/.append style={mark size=2pt}]
    \begin{axis}[
        height=\columnwidth/2.4,
        width=\columnwidth/1.9,
        ylabel={\it ACC},
        xmin=0.5, xmax=9.5,
        ymin=0.35, ymax=0.94,
        xtick={1,3,5,7,9},
        ytick={0.4,0.5,0.6,0.7,0.8,0.9},
        xticklabel style = {font=\footnotesize},
        yticklabel style = {font=\footnotesize},
        xticklabels={0.1,0.3,0.5,0.7,0.9},
        yticklabels={0.4,0.5,0.6,0.7,0.8,0.9},
        every axis y label/.style={font=\footnotesize,at={(current axis.north west)},right=2mm,above=0mm},
        legend style={fill=none,font=\small,at={(0.02,0.99)},anchor=north west,draw=none},
    ]
    \addplot[line width=0.3mm, mark=triangle]  
        plot coordinates {
(1,	0.396	)
(2,	0.406	)
(3,	0.441	)
(4,	0.473	)
(5,	0.54	)
(6,	0.583	)
(7,	0.606	)
(8,	0.616	)
(9,	0.671	)
    };

    \addplot[line width=0.3mm, mark=diamond]  
        plot coordinates {
(1,	0.513	)
(2,	0.52	)
(3,	0.521	)
(4,	0.527	)
(5,	0.521	)
(6,	0.619	)
(7,	0.551	)
(8,	0.544	)
(9,	0.566	)

    };

    \addplot[line width=0.3mm, mark=o]  
        plot coordinates {
(1,	0.916	)
(2,	0.916	)
(3,	0.911	)
(4,	0.932	)
(5,	0.912	)
(6,	0.931	)
(7,	0.91	)
(8,	0.91	)
(9,	0.89	)

    };

    \addplot[line width=0.3mm, mark=square]  
        plot coordinates {
(1,	0.36	)
(2,	0.361	)
(3,	0.361	)
(4,	0.36	)
(5,	0.359	)
(6,	0.353	)
(7,	0.37	)
(8,	0.417	)
(9,	0.444	)

    };

    \end{axis}
\end{tikzpicture}\hspace{0mm}\label{fig:vary-alpha}%
}
\subfloat[Varying $\gamma$]{
\begin{tikzpicture}[scale=1,every mark/.append style={mark size=2pt}]
    \begin{axis}[
        height=\columnwidth/2.4,
        width=\columnwidth/1.9,
        ylabel={\it ACC},
        xmin=0.5, xmax=11.5,
        ymin=0.3, ymax=0.94,
        xtick={1,3,5,7,9,11},
        ytick={0.4,0.5,0.6,0.7,0.8,0.9},
        xticklabel style = {font=\footnotesize},
        yticklabel style = {font=\footnotesize},
        xticklabels={0,2,4,6,8,10},
        yticklabels={0.4,0.5,0.6,0.7,0.8,0.9},
        every axis y label/.style={font=\footnotesize,at={(current axis.north west)},right=2mm,above=0mm},
        legend style={fill=none,font=\small,at={(0.02,0.99)},anchor=north west,draw=none},
    ]
    \addplot[line width=0.3mm, mark=triangle]  
        plot coordinates {
(1,	0.332	)
(2,	0.52	)
(3,	0.581	)
(4,	0.606	)
(5,	0.611	)
(6,	0.617	)
(7,	0.649	)
(8,	0.657	)
(9,	0.659	)
(10,	0.666	)
(11,	0.671	)
    };

    \addplot[line width=0.3mm, mark=diamond]  
        plot coordinates {
(1,	0.518	)
(2,	0.523	)
(3,	0.533	)
(4,	0.53	)
(5,	0.614	)
(6,	0.614	)
(7,	0.625	)
(8,	0.626	)
(9,	0.628	)
(10,	0.626	)
(11,	0.626	)

    };

    \addplot[line width=0.3mm, mark=o]  
        plot coordinates {
(1,	0.895	)
(2,	0.917	)
(3,	0.912	)
(4,	0.889	)
(5,	0.931	)
(6,	0.931	)
(7,	0.931	)
(8,	0.931	)
(9,	0.931	)
(10,	0.931	)
(11,	0.931	)

    };

    \addplot[line width=0.3mm, mark=square]  
        plot coordinates {
(1,	0.36	)
(2,	0.359	)
(3,	0.353	)
(4,	0.366	)
(5,	0.378	)
(6,	0.381	)
(7,	0.38	)
(8,	0.431	)
(9,	0.436	)
(10,	0.441	)
(11,	0.444	)

    };
    \end{axis}
\end{tikzpicture}\hspace{0mm}\label{fig:vary-gamma}%
}
\subfloat[Varying $T_f$]{
\begin{tikzpicture}[scale=1,every mark/.append style={mark size=2pt}]
    \begin{axis}[
        height=\columnwidth/2.4,
        width=\columnwidth/1.9,
        ylabel={\it ACC},
        xmin=0.5, xmax=5.5,
        ymin=0.4, ymax=0.94,
        xtick={1,2,3,4,5},
        ytick={0.4,0.5,0.6,0.7,0.8,0.9},
        xticklabel style = {font=\footnotesize},
        yticklabel style = {font=\footnotesize},
        xticklabels={0,5,10,15,20},
        yticklabels={0.4,0.5,0.6,0.7,0.8,0.9},
        every axis y label/.style={font=\footnotesize,at={(current axis.north west)},right=2mm,above=0mm},
        legend style={fill=none,font=\small,at={(0.02,0.99)},anchor=north west,draw=none},
    ]
    \addplot[line width=0.3mm, mark=triangle]  
        plot coordinates {
(1,	0.664	)
(2,	0.671	)
(3,	0.672	)
(4,	0.672	)
(5,	0.629	)
    };

    \addplot[line width=0.3mm, mark=diamond]  
        plot coordinates {
(1,	0.601	)
(2,	0.612	)
(3,	0.622	)
(4,	0.625	)
(5,	0.622	)

    };

    \addplot[line width=0.3mm, mark=o]  
        plot coordinates {
(1,	0.930	)
(2,	0.931	)
(3,	0.931	)
(4,	0.931	)
(5,	0.931	)

    };

    \addplot[line width=0.3mm, mark=square]  
        plot coordinates {
(1,	0.440	)
(2,	0.443	)
(3,	0.444	)
(4,	0.444	)
(5,	0.444	)

    };

    \end{axis}
\end{tikzpicture}\hspace{0mm}\label{fig:vary-tf}%
}%
\subfloat[Varying $T_g$]{
\begin{tikzpicture}[scale=1,every mark/.append style={mark size=2pt}]
    \begin{axis}[
        height=\columnwidth/2.4,
        width=\columnwidth/1.9,
        ylabel={\it ACC},
        xmin=0.5, xmax=5.5,
        ymin=0.4, ymax=0.94,
        xtick={1,2,3,4,5},
        ytick={0.4,0.5,0.6,0.7,0.8,0.9},
        xticklabel style = {font=\footnotesize},
        yticklabel style = {font=\footnotesize},
        xticklabels={0,5,10,15,20},
        yticklabels={0.4,0.5,0.6,0.7,0.8,0.9},
        every axis y label/.style={font=\footnotesize,at={(current axis.north west)},right=2mm,above=0mm},
        legend style={fill=none,font=\small,at={(0.02,0.99)},anchor=north west,draw=none},
    ]
    \addplot[line width=0.3mm, mark=triangle]  
        plot coordinates {
(1,	0.502	)
(2,	0.662	)
(3,	0.669	)
(4,	0.672	)
(5,	0.671	)

    };

    \addplot[line width=0.3mm, mark=diamond]  
        plot coordinates {
(1,	0.478	)
(2,	0.619	)
(3,	0.619	)
(4,	0.618	)
(5,	0.627	)

    };

    \addplot[line width=0.3mm, mark=o]  
        plot coordinates {
(1,	0.76	)
(2,	0.931	)
(3,	0.931	)
(4,	0.931	)
(5,	0.931	)

    };

    \addplot[line width=0.3mm, mark=square]  
        plot coordinates {
(1,	0.401	)
(2,	0.439	)
(3,	0.443	)
(4,	0.443	)
(5,	0.443	)

    };
    \end{axis}
\end{tikzpicture}\hspace{0mm}\label{fig:vary-tg}%
}%
\subfloat[Varying $\dimU$]{
\begin{tikzpicture}[scale=1,every mark/.append style={mark size=2pt}]
    \begin{axis}[
        height=\columnwidth/2.4,
        width=\columnwidth/1.9,
        ylabel={\it ACC},
        xmin=0.5, xmax=5.5,
        ymin=0.36, ymax=0.94,
        xtick={1,2,3,4,5},
        ytick={0.4,0.5,0.6,0.7,0.8,0.9},
        xticklabel style = {font=\footnotesize},
        yticklabel style = {font=\footnotesize},
        xticklabels={16,32,64,128,256},
        yticklabels={0.4,0.5,0.6,0.7,0.8,0.9},
        every axis y label/.style={font=\footnotesize,at={(current axis.north west)},right=2mm,above=0mm},
        legend style={fill=none,font=\small,at={(0.02,0.99)},anchor=north west,draw=none},
    ]
    \addplot[line width=0.3mm, mark=triangle]  
        plot coordinates {
(1,	0.624	)
(2,	0.613	)
(3,	0.592	)
(4,	0.671	)
(5,	0.598	)

    };

    \addplot[line width=0.3mm, mark=diamond]  
        plot coordinates {
(1,	0.528	)
(2,	0.618	)
(3,	0.555	)
(4,	0.518	)
(5,	0.511	)

    };

    \addplot[line width=0.3mm, mark=o]  
        plot coordinates {
(1,	0.721	)
(2,	0.931	)
(3,	0.931	)
(4,	0.931	)
(5,	0.931	)

    };

    \addplot[line width=0.3mm, mark=square]  
        plot coordinates {
(1,	0.436	)
(2,	0.444	)
(3,	0.385	)
(4,	0.437	)
(5,	0.383	)

    };
    \end{axis}
\end{tikzpicture}\hspace{4mm}\label{fig:vary-dim}%
}%
\end{small}
 \vspace{-4mm}
\caption{Clustering accuracy when varying parameters.} \label{fig:parameter}
\vspace{0mm}
\end{figure*}
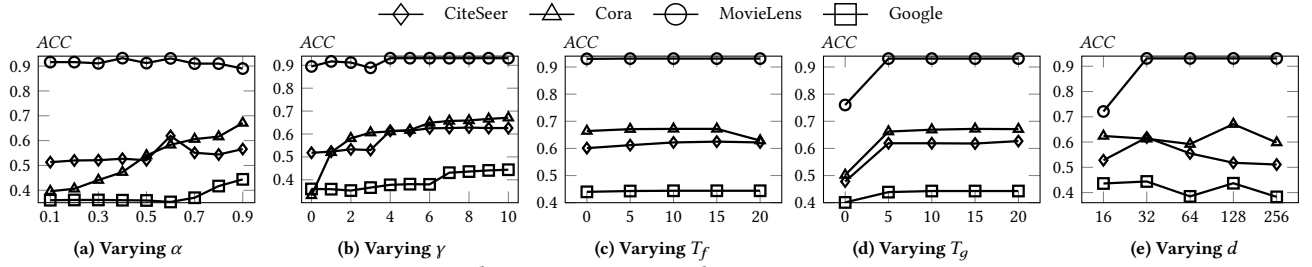

\stitle{Evaluation Metrics}
Following convention, we adopt three widely used measures \cite{lancichinetti2009detecting,cui2020adaptive,zhang2019attributed,wang2019attributed,yang2021effective,kim2022abc} to assess the clustering quality, namely (i) {\em Clustering Accuracy} (ACC), (ii) {\em Normalized Mutual Information} (NMI), and (3) {\em Adjusted Rand Index} (ARI), for measuring the quality of clusters produced by each evaluated method in the presence of the ground-truth clusters of the tested dataset. 
Particularly, ACC and NMI scores range from $0$ to $1.0$, whilst ARI ranges from $-0.5$
to $1.0$. For each of these metrics, higher values indicate better clustering quality.
Regarding efficiency evaluation, we report the running time in seconds (measured in wall-clock time) of each method on each dataset, excluding the time for input (loading datasets) and output (saving clustering results). The formulas for evaluation metrics are in Appendix \ref{sec:eval_metrics}.


\subsection{Clustering Performance}\label{sec:exp-quality}

This set of experiments reports the clustering quality achieved by \algo and all competitors over the five datasets, as well as their respective running times.
We omit a method if it cannot report the results within three days or incur out-of-memory (OOM) errors. Since \algo is randomized, we repeat it five times and report the average performance.

\stitle{Clustering Quality} Table \ref{tbl:node-clustering} shows the ACC, NMI, and ARI scores of all methods on five ABGs, and their overall average performance rankings. We highlight the top-$3$ best clustering results on each dataset in gray with darker shades indicating higher quality. \algo consistently outperforms the 17 competitors on the {\em CiteSeer}, {\em MovieLens}, and {\em Google} datasets in terms of ACC, NMI, and ARI, by substantial margins of up to $9.9\%$ for ACC, $4.5\%$ for NMI, and $12\%$ for ARI, respectively. The only exceptions are on {\em Cora} and {\em Amazon}, where \algo achieves the highest ACC and ARI results but inferior NMI scores compared to \textsf{PANE} or \textsf{AGCC}. In addition, \algo ($d=d_\U$) exhibits competitive clustering effectiveness, which either is second only to \algo or obtains the third best clustering results in most cases. Specifically, \algo ($d=d_\U$) is comparable to \algo on {\em Cora}, {\em MovieLens}, and {\em Amazon} with a performance degradation at most $0.9\%$ in ACC, $1.0\%$ in NMI, and $1.3\%$ in ARI. 
Over all datasets, \algo and \algo ($d=d_\U$) attain the best and second best average performance rank (smaller rank is better), respectively. The evident superiority of \algo and \algo ($d=d_\U$) manifests the accuracy of our proposed MSA model in Section \ref{sec:MSA} in preserving the attribute similarity and topological connections between nodes, as well as the effectiveness of theoretically-grounded three-phase optimization framework introduced in Section \ref{sec:algo}. 

At this point, a keen reader may wonder why \algo with attribute dimension reduction outperforms \algo ($d=d_\U$) on most datasets, especially {\em CiteSeer} and {\em Google}, as it seems that the former is an approximate version of the latter. Notice that \algo and \algo ($d=d_\U$) output identical results, as dimension reduction is not needed on {\em MovieLens} and \algo turns to be \algo ($d=d_\U$). Recall that the only difference between \algo and \algo ($d=d_\U$) is that \algo employs a truncated SVD over the input attribute vectors $\XM_\U$ of a node in $\U$ for dimension reduction as stated in Section \ref{sec:d}. Aside from the crucial theoretical assurance offered by this SVD-based approach in the MSA approximation, it implicitly conducts a PCA on the attribute vectors, extracting key features from the input attributes while eradicating noisy ones. In brief, the SVD-based trick in Section \ref{sec:d} grants \algo the additional ability to denoise the attribute data, thus elevating the results' quality.

\stitle{Efficiency}
For clarity, we compare the empirical efficiency of \algo and \algo ($d=d_\U$) only against competitors ranked in the top 7 for clustering quality, as shown in Table \ref{tbl:node-clustering}.
Figure \ref{fig:time} plots the computation times required by each of these methods on {\em Cora}, {\em MovieLens}, {\em Google}, and {\em Amazon}. The $y$-axis is the running time (seconds) in the log scale. On each of the diagrams in Figures \ref{fig:time-cora}, \ref{fig:time-movielens}, \ref{fig:time-google}, and \ref{fig:time-amazon}, all the bars are displayed from left to right in an ascending order w.r.t. their average performance rank in Table \ref{tbl:node-clustering}. Accordingly, except the first two bars from the left for \algo and \algo ($d=d_\U$), the third bars (from the left) in these figures illustrate the running times of the best competitors, i.e., \textsf{AGCC} on {\em Cora}, and \textsf{PANE} on {\em MovieLens}, {\em Google}, and {\em Amazon}, respectively. As we can see, \algo is consistently faster than the state-of-the-art approaches, \textsf{AGCC} or \textsf{PANE}, on four datasets, often by orders of magnitude. For instance, on {\em Cora}, {\em Google}, and {\em Amazon}, \algo takes $0.47$, $28.7$, and $178$ seconds, respectively, whereas the best baselines \textsf{AGCC} or \textsf{PANE} cost around $19$ seconds, $23$ minutes, and $4.1$ hours, respectively, attaining $40\times$, $48\times$, and $83\times$ speedup. In addition, \algo also enjoys a considerable efficiency gain of up to $19.9\times$ over \algo ($d=d_\U$), attributed to the SVD-based dimension reduction (Section \ref{sec:d}). On the {\em MovieLens} dataset, the input attribute dimension $d_\U=30$ is low, and the attribute dimension reduction is therefore disabled, making \algo and \algo ($d=d_\U$) yield the same running time, which is $3.46\times$ over the best competitor \textsf{PANE}. Although \textsf{NMF}, \textsf{KMeans}, and \textsf{SCC} run much faster than \algo on some datasets by either neglecting the graph topology or discarding the attribute data, their result quality is no match for our solution \algo.

In summary, \algo consistently delivers superior results for $k$-ABGC tasks over ABGs with various volumes while offering high practical efficiency,
which corroborates the efficacy of our novel objective function based on MSA in Section \ref{sec:preliminary} and the optimization solver with careful algorithmic designs developed in Section \ref{sec:algo}.

\subsection{Parameter Analysis}\label{sec:exp-param}
In these experiments, we empirically investigate the impact of five key parameters in \algo: $\alpha, \gamma, T_f, T_g$, and $\dimU$. For each of them, we run \algo over {\em CiteSeer}, {\em Cora}, {\em MovieLens}, and {\em Google}, respectively, by varying the parameter with others fixed as in Section \ref{sec:exp-setup}.

\stitle{Varying $\alpha$ and $\gamma$} 
\eat{
As discussed in Section \ref{sec:MSA}, the $\alpha$ is a coefficient used to balance the attribute and topological information encoded in the node feature vectors $\ZM_\U$ in Eq . \eqref{eq:PX} (the larger $\alpha$ is, the less important the role of attributes is), and $\gamma$ signifies that for each node $u_i\in \U$, its feature vector $\ZM_\U[i]$ can aggregate attributes from nodes that are at most $\gamma$ hops away from $u_i$. 
}
Figures \ref{fig:vary-alpha} 
shows that on {\em Cora} and {\em Google}, \algo's clustering performance markedly improves as $\alpha$ increases from $0.1$ to $0.9$, indicating the importance of graph structure in these datasets. On {\em CiteSeer} and {\em MovieLens}, setting $\alpha=0.6$ will be a favorable choice, which results in an optimal combination of attributes and graph topology and hence the highest ACC 
scores. Figures \ref{fig:vary-gamma} 
depicts the ACC 
scores when $\gamma$ increases from $0$ to $10$. When $\gamma=0$, the graph structure is disregarded in \algo, namely $\ZM_\U=\XM_\U$. It can be observed on all datasets that the clustering quality rises with $\gamma$ increasing except {\em CiteSeer} and {\em MovieLens}, where the ACC 
results reach a plateau after $\gamma\ge 6$. This is consistent with the fact that a larger $\gamma$ produces a more accurate solution $\ZM_\U$ to the objective in Eq. \eqref{eq:Z-obj}, and thus, higher clustering quality.

\stitle{Varying $T_f$ and $T_g$} Figures \ref{fig:vary-tf} 
presents the ACC 
scores when the $T_f$ of iterations in Algorithm \ref{alg:onmf} is varied from $0$ to $20$. We can conclude that our greedy seeding strategy described in Section \ref{sec:ONMF} is highly effective in enabling swift convergence, as additional optimization iterations merely bring minor gains in clustering performance. On {\em Cora} and {\em CiteSeer}, the ACC 
scores see an uptick when varying $\gamma$ from $0$ to $10$, followed by a pronounced downturn. Such a performance decline is caused by overfitting in solving Eq. \eqref{eq:obj4}.
From Figures \ref{fig:vary-tg} 
reporting clustering performance changes when varying $T_g$ from $0$ to $20$, we can make analogous observations on the four datasets. The evaluation scores first experience a sharp increase as $T_g$ increases from $0$ to $5$. After that, the ACC 
remain invariant with $\gamma$ increasing. The results manifest the effectiveness of our solver developed in Section \ref{sec:NCI-gen} in fast NCI generation.

\stitle{Varying $\dimU$} 
Intuitively, a large $d$ may lead to accurate preservation of MSA as per Lemma \ref{lem:FF-ZZ} and further improve clustering quality. However, in practice, the original attribute vectors $\XM_\U$ embody noises, especially when $d_\U$ is high.
As pinpointed and validated in Sections \ref{sec:d} and  \ref{sec:exp-quality}, our SVD-based dimension reduction inherently applies a PCA over $\XM_\U$ for noise elimination, considerably upgrading the empirical result quality.
That is to say, the choice of $d$ strikes a balance between capturing MSA and removing noisy data, consistent with our empirical results in Figure \ref{fig:vary-dim}.
In particular, on {\em CiteSeer}, {\em Cora}, and {\em Google}, picking $32$, $128$, and $32$ for dimension $d$, respectively, can strike a good balance between MSA preservation and noisy reduction for superior clustering performance. On {\em MovieLens}, the attribute dimension reduction is not enabled when $d\ge 32$ since its original dimension $d_\U=30$.
We refer interested readers to Appendix \ref{sec:exp-param-efficiency} for NMI and efficiency results.
\section{Related Work}\label{sec:relatedwork}
\eat{In the sequel, we review existing studies germane to $k$-ABGC and highlight where they differ from what we investigate here.}

\stitle{Bipartite Graph Clustering}
A classic methodology \cite{zhou2007bipartite} for bipartite graph clustering first projects a bipartite graph $\G$ into a unipartite graph by connecting every two nodes from the same partition $\U$ if they share common neighbors in $\G$. Then, a standard graph clustering algorithm for node clustering can be adopted on the constructed unipartite graph. However, the projection often leads to unipartite graphs $O(|\U|^2)$ edges, which is intolerable for even medium-sized graphs. In our previous work~\cite{YangShi23}, we address this problem by transforming it into a two-stage approximation framework.

Unlike the projection-based methods, another line of research focuses on simultaneously clustering two disjoint sets of nodes (i.e., $\U$ and $\V$) in a bipartite graph. These co-clustering techniques have been extensively investigated in the literature \cite{govaert2013co} and span a variety of applications in bioinformatics and text mining.
Several attempts \cite{dhillon2001co,kluger2003spectral,labiod2011co} are made to extend spectral clustering to bipartite graphs. Analogously, \citet{ailem2015co} and \citet{dhillon2003information} propose generating co-clusters by extending and optimizing classic metrics of modularity and mutual information on bipartite graphs, respectively.
\textsf{DeepCC} \cite{xu2019deep} creates low-dimension instances and features using a deep autoencoder, then assigns clusters using a variant of the Gaussian mixture model.
To handle the resolution limit in prior works as well as incorporate attribute information, \citet{kim2022abc} designed \textsf{ABC}, which incurs a severe efficiency issue due to its quadratic running time $O(|\U|^2+|\V|^2)$.

\stitle{Attributed Graph Clustering}
As surveyed in \cite{bothorel2015clustering,chunaev2019community,yang2021effective,li2023efficient}, there is a large body of work on attributed graph clustering (AGC).
According to \cite{yang2021effective}, existing AGC techniques can be categorized into four groups: edge-weigh-based methods \cite{neville2003clustering,ruan2013efficient}, distance-based methods \cite{zhou2009graph,falih2017anca}, statistics-based models \cite{yang2009combining,xu2012model,yang2021effective}, and graph learning-based methods \cite{wang2019attributed,zhang2019attributed,fanseu2023grace,wang2017mgae,mrabah2022rethinking}.
Among them, graph learning-based approaches \cite{zhang2019attributed,fanseu2023grace,mrabah2022rethinking,Dink-Net} have achieved state-of-the-art performance, as reported in \cite{yang2021effective,lai2023re}.
\eat{
More specifically, \textsf{AGCC} \cite{zhang2019attributed} employs an adaptive graph convolution to learn embeddings and then utilizes spectral clustering on the learned embeddings for clustering.
\textsf{GRACE} \cite{fanseu2023grace} performs a graph convolution on denoised nodes and directly conducts clustering, guided by its compactness measure. 
To tackle the issues of feature randomness and feature drift in AGC, \citet{mrabah2022rethinking} developed a new graph auto-encoder model \textsf{R-GAE} to perform joint clustering and embedding learning.
}
These methods obtain high clustering quality on attributed graphs at the cost of costly neural network training, thus incurring poor scalability on large graphs.
To our knowledge, the statistical-model-based solution, \textsf{ACMin} \cite{yang2021effective}, is the only AGC method that scales to massive graphs with millions of nodes and billions of edges, while attaining high result quality.
However, none of them are custom-made for ABGs, producing compromised result quality for $k$-ABGC. 

\stitle{Network Embedding}
In recent years, network embedding, which converts each node in a graph into an embedding vector capturing the surrounding structures, has been employed in a wide range of graph analytics tasks, and has seen remarkable success \cite{cui2018survey,giamphy2023survey}. In particular, by simply feeding them into data clustering methods, e.g., \textsf{KMeans}, such embedding vectors can be utilized to cope with $k$-ABGC. However, the majority of network embedding works \cite{perozzi2014deepwalk,grover2016node2vec,tang2015line,tsitsulin2018verse,qiu2019netsmf,yang2020homogeneous,gao2018bine,yang2022scalable} are designed for graphs in the absence of node attributes. To bridge this gap, a series of efforts \cite{yang2020scaling,yang2023pane,pan2021unsupervised,velivckovic2018deep,huang2017accelerated,liu2018content,wang2018united,cui2020adaptive} have been made towards incorporating node attributes into embedding vectors for enhanced result utility. These approaches still suffer from sub-optimal clustering performance as they fall short of preserving the hidden semantics underlying bipartite graphs.
To learn effective node embeddings over ABGs, \cite{huang2020biane,ahmed2020node} extend SkipGram models \cite{mikolov2013efficient} to ABGs by picking node-pair samples with consideration of both their intra-partition/inter-partition proximities and attribute similarities. \citet{athar2023asbine} project the ABG into two homogeneous graphs based on topological connections and attribute similarities, and then invoke unsupervised GNNs on the constructed graphs for embedding generation. Moreover, \citet{zhang2017learning} propose \textsf{IGE} \cite{zhang2017learning} for learning node embeddings on dynamic ABGs with a focus on temporal dependence of edges rather than the bipartite graph structures. These works either fall short of preserving multi-hop relationships between nodes or struggle to cope with large ABGs due to the significant expense of training.

\section{Conclusion}\label{sec:conclude}
This paper presents \algo, an effective and efficient solution for $k$-ABGC tasks. \algo achieves remarkable performance, attributed to a novel problem formulation based on the proposed multi-scale attribute affinity measure for nodes in ABGs, and a well-thought-out three-phase optimization framework for solving the problem.
Through a series of theoretically-grounded efficiency techniques developed in this paper, \algo is able to scale to large ABGs with millions of nodes and hundreds of millions of edges while offering state-of-the-art result quality. The superiority of \algo over 19 baselines is experimentally validated over 5 real ABGs in terms of both clustering quality and empirical efficiency.
\eat{
Regarding future work, we intend to scale \algo to larger datasets by parallelizing it over multi-core CPUs and GPUs.}

\begin{acks}
Renchi Yang is supported by the NSFC YSF grant (No. 62302414) and Hong Kong RGC ECS grant (No. 22202623). Qichen Wang is supported by Hong Kong RGC Grants (Project No. C2004-21GF and C2003-23Y). Tsz Nam Chan is supported by the NSFC grant 62202401. Jieming Shi is supported by Hong Kong RGC ECS (No. 25201221) and NSFC 62202404.
\end{acks}

\pagebreak


\appendix
\section{Appendix}\label{sec:appendix}

\subsection{Illustrative Examples}
Figure \ref{fig:ABGC} exemplifies an ABG $\G$ with 7 researchers $u_1$-$u_7$ in $\U$, 8 research publications $v_1$-$v_8$ in $\V$, and the authorships in $\EDG$. Additionally, each researcher possesses a collection of attributes, including nationality, work institution, and academic qualifications. In Figure \ref{fig:ABGC}, $u_1,u_2$ share identical attributes and close collaboration, indicating they should be grouped in the same cluster. Similar observations can be made for node pairs $(u_3,u_4)$ and $(u_6,u_7)$, where the difference is that their attributes are partially analogous. Despite limited connections related to $u_5$, $u_5$ and $u_4$ are likely to be grouped together due to their identical attributes, with $u_4$ being the sole collaborator of $u_5$. Overall, given $k=3$, an intuitive and ideal solution for $k$-ABGC in Definition \ref{def:kabgc} is to divide the 7 researchers into 3 clusters, i.e., $\C_1=\{u_1,u_2\}$, $C_2=\{u_3,u_4,u_5\}$, and $\C_3=\{u_6,u_7\}$, with considering both their connectivity (collaboration) in $\G$ and attribute homogeneity.

Example \ref{example} provides an intuitive understanding of our optimization objective function in Eq. \eqref{sec:obj}. 

\begin{example}[\bf A Running Example]\label{example} Suppose that the ABG $\G$ in Figure \ref{fig:ABGC} is unweighted, i.e., all edge weights $w(u_i,u_j)$ are 1. During preprocessing, the attributes of nodes in $\U$ are converted into 3-dimensional vectors $\XM_\U$ as in Figure \ref{fig:XU}. Assume that the parameters $\alpha$ and $\gamma$ in Eq. \eqref{eq:PX} are $0.5$ and $5$, respectively, and the number $k$ of clusters is $3$. Figure \ref{fig:ZU} displays the feature vectors for researchers $u_1$ to $u_7$ obtained by adopting the attribute aggregation in Eq. \eqref{eq:PX} and imposing the normalization in Eq. \eqref{eq:Z}. We obtain the MSA values of every two researchers in Figure \ref{fig:SU} using Eq. \eqref{eq:s}. From Figure \ref{fig:SU}, it can be observed that the nodes with the highest MSA w.r.t. $u_1$-$u_7$ (excluding themselves) are $u_2$, $u_1$, $u_4$, $u_5$, $u_4$, $u_7$, $u_6$, respectively. This implies a partition of researchers $u_1$-$u_7$ into: $\C_1={u_1,u_2}$, $\C_2={u_3,u_4,u_5}$, and $\C_3={u_6,u_7}$, optimizing the objective in Eq. \eqref{eq:obj1} (i.e., a maximization of the average intra-cluster MSA and a minimization of the average inter-cluster MSA).
\begin{figure}[H]
\vspace{-4mm}
\centering
\begin{small}
\subfloat[$\XM_\U$]{
\begin{tikzpicture}
\begin{axis}[
    axis line style={draw=none},
    width=2.5cm,
    height=3.6cm,
    colorbar,
    colorbar/width=2.5mm,
    colormap={blackwhite}{gray(0cm)=(0.9); gray(1cm)=(0.4)},
    ytick={0,1,2,3,4,5,6},
    yticklabels={$u_1$, $u_2$, $u_3$, $u_4$, $u_5$, $u_6$, $u_7$}
]
\addplot[matrix plot, point meta=explicit]
    coordinates {
    (0,0) [0.4] (1,0) [0.3] (2,0) [0.6]
    
    (0,1) [0.4] (1,1) [0.3] (2,1) [0.6]
    
    (0,2) [0.8] (1,2) [0.8] (2,2) [0.9]
    
    (0,3) [0.8] (1,3) [0.6] (2,3) [0.5]
    
    (0,4) [0.8] (1,4) [0.6] (2,4) [0.5]
    
    (0,5) [0.8] (1,5) [0.8] (2,5) [0.9]
    
    (0,6) [0.4] (1,6) [0.8] (2,6) [1.0]
 }; 
\end{axis}
\end{tikzpicture}\label{fig:XU}
}%
\subfloat[$\hat{\ZM}_\U$]{
\begin{tikzpicture}
\begin{axis}[
    axis line style={draw=none},
    width=2.5cm,
    height=3.6cm,
    colorbar,
    colorbar/width=2.5mm,
    colormap={blackwhite}{gray(0cm)=(0.9); gray(1cm)=(0.4)},
    ytick={0,1,2,3,4,5,6},
    yticklabels={$u_1$, $u_2$, $u_3$, $u_4$, $u_5$, $u_6$, $u_7$}
]
\addplot[matrix plot, point meta=explicit]
    coordinates {
    (0,0) [0.53] (1,0) [0.43] (2,0) [0.72]
    
    (0,1) [0.54] (1,1) [0.43] (2,1) [0.72]
    
    (0,2) [0.58] (1,2) [0.54] (2,2) [0.61]
    
    (0,3) [0.67] (1,3) [0.54] (2,3) [0.52]
    
    (0,4) [0.68] (1,4) [0.54] (2,4) [0.49]
    
    (0,5) [0.52] (1,5) [0.56] (2,5) [0.64]
    
    (0,6) [0.37] (1,6) [0.59] (2,6) [0.72]
 }; 
\end{axis}
\end{tikzpicture}\label{fig:ZU}
}%
\subfloat[$s(u_i,u_j)$]{
\begin{tikzpicture}
\begin{axis}[
    axis line style={draw=none},
    width=3.5cm,
    height=3.5cm,
    colorbar,
    colorbar/width=2.5mm,
    colormap={blackwhite}{gray(0cm)=(0.9); gray(1cm)=(0.4)},
    ytick={0,1,2,3,4,5,6},
    yticklabels={$u_1$, $u_2$, $u_3$, $u_4$, $u_5$, $u_6$, $u_7$},
    xtick={0,1,2,3,4,5,6},
    xticklabels={$u_1$, $u_2$, $u_3$, $u_4$, $u_5$, $u_6$, $u_7$},
    xticklabel style={font=\tiny},
    yticklabel style={font=\scriptsize}
]
\addplot[matrix plot, point meta=explicit]
coordinates {
(0,0)	[0.146]	(1,0)	[0.146]	(2,0)	[0.143]	(3,0)	[0.14]	(4,0)	[0.139]	(5,0)	[0.143]	(6,0)	[0.143]

(0,1)	[0.146]	(1,1)	[0.145]	(2,1)	[0.143]	(3,1)	[0.141]	(4,1)	[0.14]	(5,1)	[0.143]	(6,1)	[0.143]

(0,2)	[0.142]	(1,2)	[0.142]	(2,2)	[0.146]	(3,2)	[0.146]	(4,2)	[0.145]	(5,2)	[0.143]	(6,2)	[0.142]

(0,3)	[0.14]	(1,3)	[0.141]	(2,3)	[0.144]	(3,3)	[0.146]	(4,3)	[0.147]	(5,3)	[0.143]	(6,3)	[0.138]

(0,4)	[0.139]	(1,4)	[0.14]	(2,4)	[0.144]	(3,4)	[0.147]	(4,4)	[0.147]	(5,4)	[0.142]	(6,4)	[0.137]

(0,5)	[0.143]	(1,5)	[0.143]	(2,5)	[0.144]	(3,5)	[0.143]	(4,5)	[0.142]	(5,5)	[0.147]	(6,5)	[0.146]

(0,6)	[0.143]	(1,6)	[0.143]	(2,6)	[0.142]	(3,6)	[0.138]	(4,6)	[0.137]	(5,6)	[0.146]	(6,6)	[0.148]
}; 
\end{axis}
\end{tikzpicture}\label{fig:SU}
}%
\end{small}
\vspace{-3mm}
\caption{A Running Example.} \label{fig:toy-MSA}
\end{figure}
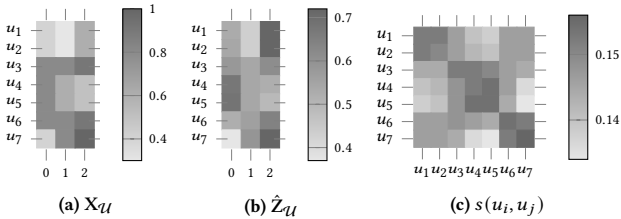

\end{example}

\eat{
\subsection{Pseudocodes of Algorithms~\ref{alg:onmf} and ~\ref{alg:round}}\label{sec:pseudo-code}
\stitle{Greedy Orthogonal NMF}
The pseudo-code of our solver introduced in Section \ref{sec:ONMF} is presented in Algorithm~\ref{alg:onmf}.

\stitle{Effective NCI Generation} 
Below we illustrate the pseudo-code of Algorithm~\ref{alg:round} for generating the NCI matrix $\YM$ in Section \ref{sec:NCI-gen}.
}

\subsection{Complexity Analysis}\label{sec:complexity}

\stitle{Algorithm~\ref{alg:aa}}
As mentioned in Section \ref{sec:d}, Lines 3-4 and Lines 7-8 in Algorithm~\ref{alg:aa} need $O(|\EDG|\cdot d\gamma+ d^3+|\U|\cdot d^2)$ time in total. As for other operations in Algorithm~\ref{alg:aa}, their processing overheads are determined by the number of entries in $\LM$, $\XM_\U$, $\widehat{\ZM}_\U$, and $\RM$, which can be bounded by $O(|\EDG|+|\U|\cdot d + d^2)$. Accordingly, the total cost entailed by Algorithm~\ref{alg:aa} is $O(|\EDG|\cdot d\gamma+ d^3+|\U|\cdot d^2)$.

\stitle{Algorithm~\ref{alg:onmf}}
Recall that the computational expense of Algorithm~\ref{alg:onmf} involves two parts: invocation of the randomized SVD (Line 1) over $\RM$ as well as $T_f$ rounds of $\FY$ and $\HM$ updates at Lines 3-5, which demands $O(|\U|\cdot d k+(\U+d)\cdot k^2)$ time by \cite{halko2011finding} and $O(|\U|\cdot d k T_f+|\U|\cdot k^2 T_f)$ time according to the analysis in Section \ref{sec:ONMF}, respectively. 

\stitle{Algorithm~\ref{alg:round}}
In Algorithm~\ref{alg:round}, the computation cost is dominated by Eq. \eqref{eq:ell-ast} ($O(|\U|\cdot k^2)$ time for all nodes per iteration) and the sparse matrix multiplication at Line 7 ($O(|\U|\cdot k)$ time per iteration), leading to $O(|\U|\cdot k^2 T_g)$ time for $T_g$ iterations in sum. 

\stitle{\algo} Overall, the asymptotic computational complexity of \algo is $O(|\EDG|\cdot d\gamma + d^3 +|\U|\cdot d^2 + |\U|\cdot d k T_f + |\U|\cdot k^2 T_g)$, which can be simplified as $O(|\EDG|+|\U|\cdot k^2)$ if $\gamma$, $d$, $T_f$ and $T_g$ are regarded as constants. In addition to the space overhead of $O(|\EDG|+|\U|\cdot d_\U)$ for the storage of $\G$, \algo requires materializing $\ZM_\U$ and $\RM$ in Algorithm~\ref{alg:aa}, amounting to a space consumption of $O(|\U|\cdot d)$. Hence, the space complexity of \algo is bounded by $O(|\EDG|+|\U|\cdot d_\U)$ since $d\le d_\U$.


\subsection{Experimental Setup Details}\label{sec:eval_metrics}

 \begin{figure}[!t]
\centering
\begin{small}
\begin{tikzpicture}
    \begin{customlegend}[legend columns=4,
        legend entries={CiteSeer, Cora, MovieLens, Google},
        legend style={at={(0.45,1.35)},anchor=north,draw=none,font=\footnotesize,column sep=0.2cm}]
    \addlegendimage{line width=0.2mm,mark size=4pt,mark=diamond}
    \addlegendimage{line width=0.2mm,mark size=4pt,mark=triangle}
    \addlegendimage{line width=0.2mm,mark size=4pt,mark=o}
    \addlegendimage{line width=0.2mm,mark size=4pt,mark=square}
    \end{customlegend}
\end{tikzpicture}
\\[-\lineskip]
\vspace{-3mm}
\subfloat[Varying $\gamma$]{
\begin{tikzpicture}[scale=1,every mark/.append style={mark size=2pt}]
    \begin{axis}[
        height=\columnwidth/2.4,
        width=\columnwidth/1.8,
        ylabel={\it running time} (sec),
        xmin=0.5, xmax=11.5,
        ymin=0.1, ymax=100,
        xtick={1,3,5,7,9,11},
        xticklabel style = {font=\footnotesize},
        yticklabel style = {font=\footnotesize},
        xticklabels={0,2,4,6,8,10},
        ymode=log,
        log basis y={10},
        every axis y label/.style={font=\footnotesize,at={(current axis.north west)},right=9mm,above=0mm},
        legend style={fill=none,font=\small,at={(0.02,0.99)},anchor=north west,draw=none},
    ]
    \addplot[line width=0.3mm, mark=triangle]  
        plot coordinates {
(1,	0.398419	)
(2,	0.404904	)
(3,	0.41543	)
(4,	0.420189	)
(5,	0.429834	)
(6,	0.435975	)
(7,	0.443174	)
(8,	0.451807	)
(9,	0.458491	)
(10,	0.489973	)
(11,	0.492316	)

    };

    \addplot[line width=0.3mm, mark=diamond]  
        plot coordinates {
(1,	0.153488	)
(2,	0.15516	)
(3,	0.158201	)
(4,	0.160443	)
(5,	0.160913	)
(6,	0.163824	)
(7,	0.163559	)
(8,	0.166855	)
(9,	0.169791	)
(10,	0.171901	)
(11,	0.172782	)

    };

    \addplot[line width=0.3mm, mark=o]  
        plot coordinates {
(1,	0.714294	)
(2,	0.806739	)
(3,	0.968287	)
(4,	1.029998	)
(5,	0.975936	)
(6,	1.414916	)
(7,	1.550673	)
(8,	1.561349	)
(9,	1.759684	)
(10,	1.856843	)
(11,	1.934207	)

    };

    \addplot[line width=0.3mm, mark=square]  
        plot coordinates {
(1,	6.837013	)
(2,	9.101928	)
(3,	11.053181	)
(4,	13.467513	)
(5,	14.911215	)
(6,	17.160328	)
(7,	19.466162	)
(8,	21.376914	)
(9,	23.421551	)
(10,	25.594696	)
(11,	27.300644	)

    };
    
    \end{axis}
\end{tikzpicture}\hspace{2mm}\label{fig:vary-gamma-time}%
}
\subfloat[Varying $T_f$]{
\begin{tikzpicture}[scale=1,every mark/.append style={mark size=2pt}]
    \begin{axis}[
        height=\columnwidth/2.4,
        width=\columnwidth/1.8,
        ylabel={\it running time} (sec),
        xmin=0.5, xmax=5.5,
        ymin=0.1, ymax=100,
        xtick={1,2,3,4,5},
        xticklabel style = {font=\footnotesize},
        yticklabel style = {font=\footnotesize},
        xticklabels={0,5,10,20,40},
        ymode=log,
        log basis y={10},
        every axis y label/.style={font=\footnotesize,at={(current axis.north west)},right=9mm,above=0mm},
        legend style={fill=none,font=\small,at={(0.02,0.99)},anchor=north west,draw=none},
    ]
    \addplot[line width=0.3mm, mark=triangle]  
        plot coordinates {
(1,	0.401663	)
(2,	0.579143	)
(3,	0.660915	)
(4,	0.683376	)
(5,	0.751038	)

    };

    \addplot[line width=0.3mm, mark=diamond]  
        plot coordinates {
(1,	0.132876	)
(2,	0.164742	)
(3,	0.19845	)
(4,	0.253895	)
(5,	0.261982	)

    };

    \addplot[line width=0.3mm, mark=o]  
        plot coordinates {
(1,	1.172929	)
(2,	1.151766	)
(3,	1.797224	)
(4,	2.017214	)
(5,	2.223069	)

    };

    \addplot[line width=0.3mm, mark=square]  
        plot coordinates {
(1,	26.257487	)
(2,	27.848746	)
(3,	28.239414	)
(4,	28.436374	)
(5,	28.580676	)

    };
    \end{axis}
\end{tikzpicture}\hspace{4mm}\label{fig:vary-tf-time}%
}%

\vspace{-2mm}
\subfloat[Varying $T_g$]{
\begin{tikzpicture}[scale=1,every mark/.append style={mark size=2pt}]
    \begin{axis}[
        height=\columnwidth/2.4,
        width=\columnwidth/1.8,
        ylabel={\it running time} (sec),
        xmin=0.5, xmax=5.5,
        ymin=0.1, ymax=100,
        xtick={1,2,3,4,5},
        xticklabel style = {font=\footnotesize},
        yticklabel style = {font=\footnotesize},
        xticklabels={0,5,10,20,40},
        ymode=log,
        log basis y={10},
        every axis y label/.style={font=\footnotesize,at={(current axis.north west)},right=9mm,above=0mm},
        legend style={fill=none,font=\small,at={(0.02,0.99)},anchor=north west,draw=none},
    ]
    \addplot[line width=0.3mm, mark=triangle]  
        plot coordinates {
(1,	0.458161	)
(2,	0.461147	)
(3,	0.47778	)
(4,	0.5031	)
(5,	0.514223	)

    };

    \addplot[line width=0.3mm, mark=diamond]  
        plot coordinates {
(1,	0.15678	)
(2,	0.159429	)
(3,	0.163092	)
(4,	0.165056	)
(5,	0.166336	)

    };

    \addplot[line width=0.3mm, mark=o]  
        plot coordinates {
(1,	1.226954	)
(2,	1.432011	)
(3,	1.560567	)
(4,	1.578336	)
(5,	1.669809	)

    };
    
    \addplot[line width=0.3mm, mark=square]  
        plot coordinates {
(1,	27.045351	)
(2,	27.161634	)
(3,	27.248344	)
(4,	27.37252	)
(5,	27.433951	)

    };
    \end{axis}
\end{tikzpicture}\hspace{2mm}\label{fig:vary-tg-time}%
}%
\subfloat[Varying $\dimU$]{
\begin{tikzpicture}[scale=1,every mark/.append style={mark size=2pt}]
    \begin{axis}[
        height=\columnwidth/2.4,
        width=\columnwidth/1.8,
        ylabel={\it running time} (sec),
        xmin=0.5, xmax=5.5,
        ymin=0.1, ymax=200,
        xtick={1,2,3,4,5},
        xticklabel style = {font=\footnotesize},
        yticklabel style = {font=\footnotesize},
        xticklabels={16,32,64,128,256},
        ymode=log,
        log basis y={10},
        every axis y label/.style={font=\footnotesize,at={(current axis.north west)},right=9mm,above=0mm},
        legend style={fill=none,font=\small,at={(0.02,0.99)},anchor=north west,draw=none},
    ]
    \addplot[line width=0.3mm, mark=triangle]  
        plot coordinates {
(1,	0.107339	)
(2,	0.234953	)
(3,	0.230472	)
(4,	0.472356	)
(5,	0.766333	)

    };

    \addplot[line width=0.3mm, mark=diamond]  
        plot coordinates {
(1,	0.106607	)
(2,	0.165781	)
(3,	0.249661	)
(4,	0.428256	)
(5,	0.883366	)

    };

    \addplot[line width=0.3mm, mark=o]  
        plot coordinates {
(1,	1.100654	)
(2,	1.587722	)
(3,	1.154994	)
(4,	1.434372	)
(5,	1.533608	)

    };

    \addplot[line width=0.3mm, mark=square]  
        plot coordinates {
(1,	11.680138	)
(2,	26.839837	)
(3,	48.084676	)
(4,	103.282308	)
(5,	180.269697	)

    };
    \end{axis}
\end{tikzpicture}\hspace{4mm}\label{fig:vary-dim-time}%
}%
\end{small}
 \vspace{-2mm}
\caption{Efficiency when varying parameters.} \label{fig:parameter-time}
\vspace{-2mm}
\end{figure}
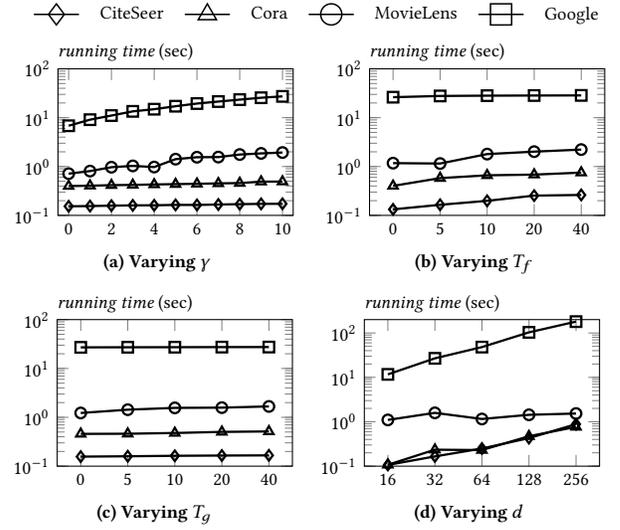

  \begin{figure*}[!t]
\centering
\begin{small}
\begin{tikzpicture}
    \begin{customlegend}[legend columns=4,
        legend entries={CiteSeer, Cora, MovieLens, Google},
        legend style={at={(0.45,1.35)},anchor=north,draw=none,font=\footnotesize,column sep=0.2cm}]
    \addlegendimage{line width=0.2mm,mark size=4pt,mark=diamond}
    \addlegendimage{line width=0.2mm,mark size=4pt,mark=triangle}
    \addlegendimage{line width=0.2mm,mark size=4pt,mark=o}
    \addlegendimage{line width=0.2mm,mark size=4pt,mark=square}
    \end{customlegend}
\end{tikzpicture}
\\[-\lineskip]
\vspace{-3mm}
\subfloat[Varying $\alpha$]{
\begin{tikzpicture}[scale=1,every mark/.append style={mark size=2pt}]
    \begin{axis}[
        height=\columnwidth/2.4,
        width=\columnwidth/1.9,
        ylabel={\it NMI},
        xmin=0.5, xmax=9.5,
        ymin=0.10, ymax=0.98,
        xtick={1,3,5,7,9},
        ytick={0.1,0.3,0.5,0.7,0.9},
        xticklabel style = {font=\footnotesize},
        yticklabel style = {font=\footnotesize},
        xticklabels={0.1,0.3,0.5,0.7,0.9},
        yticklabels={0.1,0.3,0.5,0.7,0.9},
        every axis y label/.style={font=\footnotesize,at={(current axis.north west)},right=2mm,above=0mm},
        legend style={fill=none,font=\small,at={(0.02,0.99)},anchor=north west,draw=none},
    ]
    \addplot[line width=0.3mm, mark=triangle]  
        plot coordinates {
(1,0.181)
(2,0.201)
(3,0.241)
(4,0.286)
(5,0.332)
(6,0.385)
(7,0.423)
(8,0.435)
(9,0.475)

    };

    \addplot[line width=0.3mm, mark=diamond]  
        plot coordinates {
(1,0.245)
(2,0.249)
(3,0.249)
(4,0.259)
(5,0.258)
(6,0.316)
(7,0.256)
(8,0.244)
(9,0.241)

    };

    \addplot[line width=0.3mm, mark=o]  
        plot coordinates {
(1,0.945)
(2,0.946)
(3,0.939)
(4,0.963)
(5,0.94)
(6,0.961)
(7,0.949)
(8,0.949)
(9,0.926)

    };

    \addplot[line width=0.3mm, mark=square]  
        plot coordinates {
(1,0.114)
(2,0.114)
(3,0.116)
(4,0.118)
(5,0.121)
(6,0.125)
(7,0.13)
(8,0.107)
(9,0.135)

    };

    \end{axis}
\end{tikzpicture}\hspace{0mm}\label{fig:vary-alpha-NMI}%
}
\subfloat[Varying $\gamma$]{
\begin{tikzpicture}[scale=1,every mark/.append style={mark size=2pt}]
    \begin{axis}[
        height=\columnwidth/2.4,
        width=\columnwidth/1.9,
        ylabel={\it NMI},
        xmin=0.5, xmax=11.5,
        ymin=0.10, ymax=0.98,
        xtick={1,3,5,7,9,11},
        ytick={0.1,0.3,0.5,0.7,0.9},
        xticklabel style = {font=\footnotesize},
        yticklabel style = {font=\footnotesize},
        xticklabels={0,2,4,6,8,10},
        yticklabels={0.1,0.3,0.5,0.7,0.9},
        every axis y label/.style={font=\footnotesize,at={(current axis.north west)},right=2mm,above=0mm},
        legend style={fill=none,font=\small,at={(0.02,0.99)},anchor=north west,draw=none},
    ]
    \addplot[line width=0.3mm, mark=triangle]  
        plot coordinates {
(1,0.116)
(2,0.298)
(3,0.39)
(4,0.419)
(5,0.432)
(6,0.44)
(7,0.451)
(8,0.459)
(9,0.461)
(10,0.463)
(11,0.473)

    };

    \addplot[line width=0.3mm, mark=diamond]  
        plot coordinates {
(1,0.242)
(2,0.252)
(3,0.257)
(4,0.26)
(5,0.312)
(6,0.311)
(7,0.322)
(8,0.324)
(9,0.324)
(10,0.322)
(11,0.322)

    };

    \addplot[line width=0.3mm, mark=o]  
        plot coordinates {
(1,0.948)
(2,0.947)
(3,0.94)
(4,0.926)
(5,0.962)
(6,0.961)
(7,0.961)
(8,0.961)
(9,0.961)
(10,0.961)
(11,0.96)

    };

    \addplot[line width=0.3mm, mark=square]  
        plot coordinates {
(1,0.113)
(2,0.118)
(3,0.124)
(4,0.128)
(5,0.131)
(6,0.132)
(7,0.132)
(8,0.118)
(9,0.125)
(10,0.131)
(11,0.136)

    };
    \end{axis}
\end{tikzpicture}\hspace{0mm}\label{fig:vary-gamma-NMI}%
}
\subfloat[Varying $T_f$]{
\begin{tikzpicture}[scale=1,every mark/.append style={mark size=2pt}]
    \begin{axis}[
        height=\columnwidth/2.4,
        width=\columnwidth/1.9,
        ylabel={\it NMI},
        xmin=0.5, xmax=5.5,
        ymin=0.10, ymax=0.98,
        xtick={1,2,3,4,5},
        ytick={0.1,0.3,0.5,0.7,0.9},
        xticklabel style = {font=\footnotesize},
        yticklabel style = {font=\footnotesize},
        xticklabels={0,5,10,15,20},
        yticklabels={0.1,0.3,0.5,0.7,0.9},
        every axis y label/.style={font=\footnotesize,at={(current axis.north west)},right=2mm,above=0mm},
        legend style={fill=none,font=\small,at={(0.02,0.99)},anchor=north west,draw=none},
    ]
    \addplot[line width=0.3mm, mark=triangle]  
        plot coordinates {
(1,0.465)
(2,0.475)
(3,0.476)
(4,0.419)
(5,0.274)

    };

    \addplot[line width=0.3mm, mark=diamond]  
        plot coordinates {
(1,0.30)
(2,0.309)
(3,0.32)
(4,0.320)
(5,0.229)

    };

    \addplot[line width=0.3mm, mark=o]  
        plot coordinates {
(1,0.960)
(2,0.961)
(3,0.961)
(4,0.961)
(5,0.961)

    };

    \addplot[line width=0.3mm, mark=square]  
        plot coordinates {
(1,0.135)
(2,0.135)
(3,0.135)
(4,0.135)
(5,0.135)

    };

    \end{axis}
\end{tikzpicture}\hspace{0mm}\label{fig:vary-tf-NMI}%
}%
\subfloat[Varying $T_g$]{
\begin{tikzpicture}[scale=1,every mark/.append style={mark size=2pt}]
    \begin{axis}[
        height=\columnwidth/2.4,
        width=\columnwidth/1.9,
        ylabel={\it NMI},
        xmin=0.5, xmax=5.5,
        ymin=0.10, ymax=0.98,
        xtick={1,2,3,4,5},
        ytick={0.1,0.3,0.5,0.7,0.9},
        xticklabel style = {font=\footnotesize},
        yticklabel style = {font=\footnotesize},
        xticklabels={0,5,10,15,20},
        yticklabels={0.1,0.3,0.5,0.7,0.9},
        every axis y label/.style={font=\footnotesize,at={(current axis.north west)},right=2mm,above=0mm},
        legend style={fill=none,font=\small,at={(0.02,0.99)},anchor=north west,draw=none},
    ]
    \addplot[line width=0.3mm, mark=triangle]  
        plot coordinates {
(1,0.316)
(2,0.472)
(3,0.472)
(4,0.475)
(5,0.475)

    };

    \addplot[line width=0.3mm, mark=diamond]  
        plot coordinates {
(1,0.225)
(2,0.314)
(3,0.316)
(4,0.323)
(5,0.312)

    };

    \addplot[line width=0.3mm, mark=o]  
        plot coordinates {
(1,0.789)
(2,0.961)
(3,0.961)
(4,0.961)
(5,0.961)

    };

    \addplot[line width=0.3mm, mark=square]  
        plot coordinates {
(1,0.108)
(2,0.131)
(3,0.134)
(4,0.135)
(5,0.135)

    };
    \end{axis}
\end{tikzpicture}\hspace{0mm}\label{fig:vary-tg-NMI}%
}%
\subfloat[Varying $\dimU$]{
\begin{tikzpicture}[scale=1,every mark/.append style={mark size=2pt}]
    \begin{axis}[
        height=\columnwidth/2.4,
        width=\columnwidth/1.9,
        ylabel={\it NMI},
        xmin=0.5, xmax=5.5,
        ymin=0.10, ymax=0.98,
        xtick={1,2,3,4,5},
        ytick={0.1,0.3,0.5,0.7,0.9},
        xticklabel style = {font=\footnotesize},
        yticklabel style = {font=\footnotesize},
        xticklabels={16,32,64,128,256},
        yticklabels={0.1,0.3,0.5,0.7,0.9},
        every axis y label/.style={font=\footnotesize,at={(current axis.north west)},right=2mm,above=0mm},
        legend style={fill=none,font=\small,at={(0.02,0.99)},anchor=north west,draw=none},
    ]
    \addplot[line width=0.3mm, mark=triangle]  
        plot coordinates {
(1,0.432)
(2,0.438)
(3,0.429)
(4,0.475)
(5,0.421)

    };

    \addplot[line width=0.3mm, mark=diamond]  
        plot coordinates {
(1,0.251)
(2,0.317)
(3,0.268)
(4,0.236)
(5,0.229)

    };

    \addplot[line width=0.3mm, mark=o]  
        plot coordinates {
(1,0.727)
(2,0.961)
(3,0.961)
(4,0.961)
(5,0.961)

    };

    \addplot[line width=0.3mm, mark=square]  
        plot coordinates {
(1,0.123)
(2,0.136)
(3,0.134)
(4,0.126)
(5,0.124)

    };
    \end{axis}
\end{tikzpicture}\hspace{4mm}\label{fig:vary-dim-NMI}%
}%
\end{small}
 \vspace{-2mm}
\caption{Clustering quality (NMI) when varying parameters.} \label{fig:parameter-NMI}
\vspace{0mm}
\end{figure*}
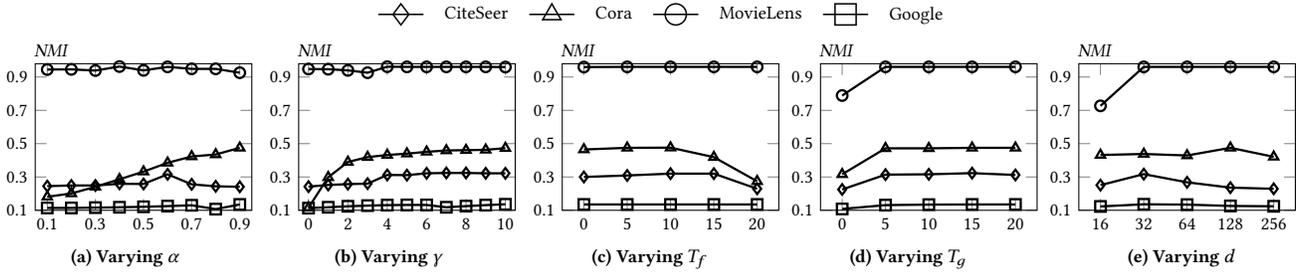

\stitle{Implementation Details}
For \textsf{KMeans}, \textsf{NMF}, \textsf{SpecClust}, \textsf{SCC}, and \textsf{SBC}, we use their standard implementations from the machine learning library scikit-learn \cite{sklearn_api} with default parameters. We also adopt the implementations of \textsf{InfoCC}, \textsf{SpecMOD}, and \textsf{CCMOD} from the Coclust package \cite{role2019coclust}.
As for the rest of the competitors, we collect their source codes from the respective authors and use the parameter settings suggested in their papers.
The data clustering methods and bipartite graph co-clustering algorithms are conducted on the $\XM_\U$ and the graph data (e.g., adjacency matrix) without $\XM_\U$, respectively, owing to their inherent designs.
The codes of all competitors are collected from their respective authors or popular open-source libraries, and all are implemented in Python. 

\stitle{Evaluation Metrics}
The specific mathematical definitions of {\em Clustering Accuracy} (ACC), {\em Normalized Mutual Information} (NMI), and {\em Adjusted Rand Index} (ARI) are as follows:
\begin{small}
\begin{equation*}
ACC = \frac{\sum_{u_i\in \U}{\mathbb{1}_{y_{u_i}=map(y^{\prime}_{u_i})}}}{|\U|},
\end{equation*}  
\end{small}
where $y^{\prime}_{u_i}$ and $y_{u_i}$ stand for the predicted and ground-truth cluster labels of node $u_i$, respectively, $map(y^{\prime}_{u_i})$ is the permutation function that maps each $y^{\prime}_{u_i}$ to the equivalent cluster label provided via Hungarian algorithm \cite{kuhn1955hungarian}, and the value of $\mathbb{1}_{y_{u_i}=map(y^{\prime}_{u_i})}$ is 1 if $y_{u_i}=map(y^{\prime}_{u_i})$ and 0 otherwise,
\begin{small}
\begin{equation*}
NMI = \frac{\sum_{i=1}^{k}\sum_{j=1}^{k}{|\C^{\ast}_i\cap \C_j|\cdot \log{\frac{|\C^{\ast}_i\cap \C_j|}{|\C^{\ast}_i|\cdot |\C_j|}}}}{\sqrt{\sum_{i=1}^k{|\C^{\ast}_i|\cdot \log{\frac{|{\C^{\ast}_i}|}{|\U|}}}}\cdot \sqrt{\sum_{i=1}^k{|\C_i|\cdot \log{\frac{|{\C_i}|}{|\U|}}}}},\ \text{and}
\end{equation*}
\begin{equation*}
ARI=\frac{\sum_{i=1}^k\sum_{j=1}^k{\binom {|\C^{\ast}_i\cap \C_j|}2}-\left(\sum_{i=1}^k{\binom {|\C^{\ast}_i|}2}\cdot \sum_{j=1}^k{\binom {|\C_j|} 2}\right)/{\binom {|\U|}2}}{0.5\left(\sum_{i=1}^k{\binom {|\C^{\ast}_i|}2}+ \sum_{j=1}^k{\binom {|\C_j|}2} \right)- \left(\sum_{i=1}^k{\binom {|\C^{\ast}_i|}2}\cdot \sum_{j=1}^k{\binom {|\C_j|}2}\right)/{\binom {|\U|}2} },
\end{equation*}
\end{small}
where $\C^{\ast}_i$ and $\C_i$ represent the $i$-th ground-truth and predicted clusters for $\U$ in $\G$, respectively.

\subsection{Parameter Analysis in Efficiency}\label{sec:exp-param-efficiency}

Figure \ref{fig:parameter-time} illustrates the running time of \algo when varying parameters $\gamma$, $T_f$, $T_g$, and $d$. It can be observed that $d$ is the most impactful parameter on the computational time of \algo as its time complexity is proportional to $d^2$ analyzed in Section \ref{sec:complexity}. Regarding $\gamma$ and $T_f$, they affect the efficiency of Algorithms \ref{alg:aa} and \ref{alg:onmf}, respectively, which engender slight runtime growth. In contrast, Figure \ref{fig:vary-tg-time} shows that the running time of \algo almost stays steady, demonstrating the superb efficiency of Algorithm \ref{alg:round} compared to the other two procedures.



Figure \ref{fig:parameter-NMI} presents the NMI scores when varying parameters of \algo as in Section \ref{sec:exp-param}, which are quantitatively similar to the accuracy results in Figure \ref{fig:parameter}.


\subsection{Theorems and Proofs}\label{sec:proofs}
\begin{theorem}[\bf Eckart–Young Theorem \cite{gloub1996matrix}]\label{lem:eym}
Suppose that $\MM_{k}\in\mathbb{R}^{n\times k}$ is the rank-$k$ approximation to $\MM\in\mathbb{R}^{n\times n}$ obtained by exact SVD, then
$\min_{rank(\widehat{\MM})\le k}{\|\MM-\widehat{\MM}\|_2}=\|\MM-\MM_{k}\|_2=\sigma_{k+1}$,
where $\sigma_{i}$ represents the $i$-th largest singular value of $\MM$.
\end{theorem}

\begin{proof}[\bf Proof of Lemma \ref{lem:ZPX}]
We first rewrite that Eq. \eqref{eq:Og} as 
\begin{equation}\label{eq:ZLLZ}
\Tr(\ZM_\U^{\top}\cdot(\IM-\LM_\U\LM_\U^{\top})\cdot\ZM_\U).
\end{equation}
The equivalence can be deduced by the definition of Eq. \eqref{eq:Og},
\begin{small}
\begin{align*}
&\OO_g=\frac{1}{2}\sum_{u_i,u_j\in \U}{{\widehat{w}(u_i,u_j)}\cdot \left\|\frac{\ZM_\U[i]}{\sqrt{\DM_\U[i,i]}}-\frac{\ZM_\U[j]}{\sqrt{\DM_\U[j,j]}}\right\|^2}\\
&=\sum_{u_i,u_j\in \U}{ \frac{\widehat{w}(u_i,u_j)}{2} \left(\frac{\|\ZM_\U[i]\|^2}{{\DM_\U[i,i]}}+\frac{\|\ZM_\U[j]\|^2}{{\DM_\U[j,j]}} -\frac{2\ZM_\U[i]\cdot\ZM_\U[j]}{\sqrt{\DM_\U[i,i]}\sqrt{\DM_\U[j,j]}} \right)} \\
&= \sum_{v_\ell\in \V}{\sum_{u_i,u_j\in \U}}{\frac{w(u_i,v_\ell) w(u_j,v_\ell)}{\DM_\V[\ell,\ell]} \left(\frac{\|\ZM_\U[i]\|^2}{\DM_\U[i,i]} - \frac{\ZM_\U[i]\cdot\ZM_\U[j]}{\sqrt{\DM_\U[i,i]}\sqrt{\DM_\U[j,j]}}\right)}\\
&= \sum_{u_i\in \U}{\|\ZM_\U[i]\|^2}- \sum_{v_\ell\in \V}{\sum_{u_i,u_j\in \U} \LM_\U[i,\ell]\cdot \LM[j,\ell]\cdot \ZM_\U[i]\cdot\ZM_\U[j]} \\
&=\sum_{u_i}{\ZM_\U[i]\cdot \ZM_\U[i]} - \sum_{v_\ell\in \V}{(\LM_\U^{\top}\ZM_\U)[\ell]\cdot (\LM_\U^{\top}\ZM_\U)[\ell]}  \\
&=\Tr(\ZM_\U^{\top}\cdot(\IM-\LM_\U\LM_\U^{\top})\cdot\ZM_\U).
\end{align*}
\end{small}
Next, we set the derivative of Eq. \eqref{eq:ZLLZ} w.r.t. $\ZM_\U$ to zero and get the optimal $\ZM_\U$ as:
\begin{align*}
& \frac{\partial{\{(1-\alpha)\cdot\|\ZM_\U - \XM_\U\|^2_F + \alpha \cdot \Tr(\ZM^{\top}_\U(\IM-\LM_\U\LM_\U^{\top})\ZM_\U)\}}}{\partial{\ZM_\U}}=0\\
& \Longrightarrow (1-\alpha)\cdot(\ZM_\U - \XM_\U) + \alpha (\IM-\LM_\U\LM_\U^{\top})^{\top}\ZM_\U = 0 \\
& \Longrightarrow \ZM_\U = ((1-\alpha)\cdot\IM+\alpha\cdot (\IM-\LM_\U\LM_\U^{\top}))^{-1}\cdot (1-\alpha)\XM_\U.
\end{align*}
By Neumann series, \ie $(\IM-\MM)^{-1}=\sum_{k=0}^{\infty}{\MM^k}$, we have
\begin{align*}
& \left((1-\alpha)\cdot\IM+\alpha (\IM-\LM_\U\LM_\U^{\top}) \right)^{-1} = \left((1-\alpha)\cdot\IM+ \alpha (\IM-\LM_\U\LM_\U^{\top})\right)^{-1}\\
& = \left(\IM - \alpha\cdot \LM_\U\LM_\U^{\top}\right)^{-1} = \sum_{r=0}^{\infty}{\alpha^r\cdot (\LM_\U\LM_\U^{\top})^{r}},
\end{align*}
which seals the proof.
\end{proof}

\begin{proof}[\bf Proof of Lemma \ref{lem:obj3}]
We first need the following lemma:
\begin{lemma}\label{lem:obj2}
The optimization objective in Eq. \eqref{eq:obj1} is equivalent to optimizing:
$\max_{\YM} \Tr\left(\YM^{\top} \SM \YM\right)$,
where $\YM$ is defined in Eq. \eqref{eq:Y} and $\SM$ is a $|\U|\times |\U|$ matrix where $\SM[i,j]=s(u_i,u_j)$.
\end{lemma}

Then, by the connection of the Frobenius norm and trace of matrices, i.e., $\|\MM\|^2_F=\Tr(\MM^{\top}\MM)=\Tr(\MM\MM^{\top})$, we have
\begin{align*}
\mathcal{J} = \|\RM-\YM\HM^{\top}\|^2_F = & \Tr(\RM\RM^{\top}-\RM\HM\YM^{\top}-\YM\HM^{\top}\RM^{\top}+\YM\HM^{\top}\HM\YM^{\top})\\
= & \Tr(\RM\RM^{\top}) - \Tr(\RM\HM\YM^{\top}) - \Tr(\YM\HM^{\top}\RM^{\top}) \\
& + \Tr(\YM\HM^{\top}\HM\YM^{\top}) \\
= &  \Tr(\RM\RM^{\top}-2\YM\HM^{\top}\RM^{\top}) - \Tr(\YM^{\top}\YM\HM^{\top}\HM)
\\
= & \Tr(\RM\RM^{\top}-2\YM\HM^{\top}\RM^{\top} + \HM^{\top}\HM).
\end{align*}
The zero gradient condition $\frac{\partial \mathcal{J}}{\partial \HM}=-2\RM^{\top}\YM+2\YM=0$ leads to $\HM=\RM^{\top}\YM$. Hence,
\begin{align}
\mathcal{J}=&\Tr(\RM\RM^{\top}-2\YM\HM^{\top}\RM^{\top} + \HM^{\top}\HM) \notag\\
& = \Tr(\RM\RM^{\top}) - 2\Tr(\YM\YM^{\top}\RM\RM^{\top})+\Tr(\YM^{\top}\RM\RM^{\top}\YM) \notag \\
& = \Tr(\RM\RM^{\top}) - 2\Tr(\YM^{\top}\RM\RM^{\top}\YM)+ \Tr(\YM^{\top}\RM\RM^{\top}\YM) \notag \\
& = \Tr(\RM\RM^{\top}) + \Tr(\YM^{\top}\RM\RM^{\top}\YM).\label{eq:J-YRRY}
\end{align}
Since $\Tr(\RM\RM^{\top})$ is a constant, maximizing $\mathcal{J}=\|\RM-\YM\HM^{\top}\|^2_F$ is equivalent to maximizing $\Tr(\YM^{\top}\RM\RM^{\top}\YM)=\Tr(\YM^{\top}\SM\YM)$.
\end{proof}

\begin{proof}[\bf Proof of Lemma \ref{lem:obj2}]
Let $\SM$ be a $|\U|\times |\U|$ matrix where $\SM[i,j]=s(u_i,u_j)$ and $\SM_d$ be a $|\U|\times |\U|$ diagonal matrix in which $\SM_d[i,i]=\sum_{u_j\in \U}{s(u_i,u_j)}$. Then,
\begin{normalsize}
\begin{align*}
& \frac{1}{|\C_\ell|}\sum_{u_i\in \C_\ell, u_j\in \U\setminus \C_\ell}{s(u_i,u_j)} = \frac{1}{2}\sum_{u_i,u_j\in \U}{\SM[i,j]\cdot (\YM[i,\ell]-\YM[j,\ell])^2}\\
&= \YM[:,\ell]^{\top}\cdot (\SM_d-\SM)\cdot \YM[:,\ell].
\end{align*}
\end{normalsize}
Thus, we can rewrite Eq. \eqref{eq:obj1} as $\min_{\YM}{\Tr(\YM^{\top} (\SM_d-\SM) \YM)}$.
Note that it can be further simplified as $\max_{\YM} \Tr\left(\YM^{\top} \SM \YM\right)$, which completes the proof.
\end{proof}

\begin{proof}[\bf Proof of Theorem \ref{lem:RR}]
Let $\mathbf{z}=\widehat{\ZM}_\U[i] - \widehat{\ZM}_\U[j]$ and $\widehat{\QM}=\sqrt{d_\U}\cdot \QM$. Based on the definition of $\RM^{\prime}$ in Eq. \eqref{eq:R-prime}, we can get
\begin{normalsize}
\begin{align}
\mathbb{E}[\RM^{\prime}[i]\cdot \RM^{\prime}[j]] = & \mathbb{E} \Bigg[\frac{e}{d_\U}\sum_{\ell=1}^{d_\U} sin(\widehat{\QM}[\ell]\cdot \widehat{\ZM}_\U[i])\cdot sin(\widehat{\QM}[\ell]\cdot \widehat{\ZM}_\U[j]) \notag\\
&\quad\quad  - cos(\widehat{\QM}[\ell]\cdot \widehat{\ZM}_\U[i])\cdot cos(\widehat{\QM}[\ell]\cdot \widehat{\ZM}_\U[j]) \Bigg] \notag\\
= & \mathbb{E}\left[\frac{e}{d_\U}\sum_{\ell=1}^{d_\U} cos(\widehat{\QM}[\ell]\cdot (\widehat{\ZM}_\U[i] - \widehat{\ZM}_\U[j])) \right] \notag\\
= & \mathbb{E}\left[\frac{e}{d_\U}\sum_{\ell=1}^{d_\U} cos(\widehat{\QM}[\ell]\cdot \mathbf{z}) \right] \label{eq:RR}
\end{align}
\end{normalsize}
By Lemma 5 in \cite{yu2016orthogonal}, for any vector $\widehat{\QM}[\ell]$, 
\begin{normalsize}
\begin{equation*}
\left|\frac{\mathbb{E}[cos(\widehat{\QM}[\ell]\cdot \mathbf{z})]}{e^{-{\|\mathbf{z}\|^2}/{2}}} - \left(1-\frac{\|\mathbf{z}\|^4}{4d_\U}\right) \right| \le \frac{\|\mathbf{z}\|^4\cdot (\|\mathbf{z}\|^4+8\|\mathbf{z}\|^2+8)}{16d_\U^2}.
\end{equation*}
\end{normalsize}
Note that for each $u_i\in \U$, the row vector $\widehat{\ZM}_\U[i]$ is $L_2$ normalized (Line 5), i.e., $\|\mathbf{z}\|^2\in [0,1]$. Based thereon, the above inequality can be transformed into
\begin{normalsize}
\begin{equation}\label{eq:cos-ex-exp}
1- \frac{17}{16d_\U^2} - \frac{1}{4d_\U} \le \frac{\mathbb{E}[cos(\widehat{\QM}[\ell]\cdot \mathbf{z})]}{e^{-{\|\mathbf{z}\|^2}/{2}}} \le 1 + \frac{17}{16d_\U^2} + \frac{1}{4d_\U}.
\end{equation}
\end{normalsize}
According to Line 5 of Algorithm \ref{alg:aa}, $\|\widehat{\ZM}_\U[i]\|^2=1\ \forall{u_i\in \U}$. Thus, 
\begin{align*}
\|\widehat{\ZM}_\U[i] - \widehat{\ZM}_\U[j]\|^2 &= \|\widehat{\ZM}_\U[i]\|^2 + \|\widehat{\ZM}_\U[j]\|^2 - 2\widehat{\ZM}_\U[i]\cdot \widehat{\ZM}_\U[j]\\
& =2(1-\widehat{\ZM}_\U[i]\cdot \widehat{\ZM}_\U[j]).
\end{align*}
Thus, $e\cdot e^{-\frac{\|\mathbf{z}\|^2}{2}} = e \cdot e^{ -\frac{\| \widehat{\ZM}_\U[i] - \widehat{\ZM}_\U[j] \|^2}{2}} = e^{\widehat{\ZM}_\U[i]\cdot \widehat{\ZM}_\U[j]}$.
Combining Eq. \eqref{eq:cos-ex-exp} and Eq. \eqref{eq:RR} leads to
\begin{normalsize}
\begin{equation}\label{eq:RR-ee}
1- \frac{17}{16d_\U^2} - \frac{1}{4d_\U} \le \frac{\mathbb{E}[\RM^{\prime}[i]\cdot \RM^{\prime}[j]]}{e^{\widehat{\ZM}_\U[i]\cdot \widehat{\ZM}_\U[j]}} \le 1 + \frac{17}{16d_\U^2} + \frac{1}{4d_\U}.
\end{equation}
\end{normalsize}
Further, Eq. \eqref{eq:RR-ee} implies 
\begin{normalsize}
\begin{equation}\label{eq:R-r}
1- \frac{17}{16d_\U^2} - \frac{1}{4d_\U} \le \frac{\mathbb{E}[\sum_{u_\ell\in \U}{\RM^{\prime}[i]\cdot \RM^{\prime}[\ell]}]}{\sum_{u_\ell\in \U}{ e^{\widehat{\ZM}_\U[i]\cdot \widehat{\ZM}_\U[\ell]}}} \le 1 + \frac{17}{16d_\U^2} + \frac{1}{4d_\U}.
\end{equation}
\end{normalsize}
According to Eq. \eqref{eq:Ri}, we can obtain
\begin{normalsize}
\begin{equation*}
\mathbb{E}[\RM[i]\cdot \RM[j]] = \mathbb{E}\left[ \frac{\RM^{\prime}[i]\cdot \RM^{\prime}[j]}{\sqrt{\RM^{\prime}[i]\cdot \mathbf{r}}\cdot \sqrt{\RM^{\prime}[j]\cdot \mathbf{r}}} \right].
\end{equation*}
\end{normalsize}
Note that, by $\mathbf{r}$'s definition in Eq. \eqref{eq:Ri}, we have $\mathbb{E}[ \RM^{\prime}[i]\cdot \mathbf{r}] = \mathbb{E}[\sum_{u_\ell\in \U}{\RM^{\prime}[i]\cdot \RM^{\prime}[\ell]}]$. Therefore, plugging Eq. \eqref{eq:RR-ee} and Eq. \eqref{eq:R-r} into the above equation proves the theorem.
\end{proof}

\begin{proof}[\bf Proof of Lemma \ref{lem:RYH-RGammaH}]
According to Eq. \eqref{eq:J-YRRY}, optimizing Eq. \eqref{eq:obj3} is equivalent to maximizing $\Tr(\YM^{\top}\RM\RM^{\top}\YM)$.
Using the cyclic property of matrix trace, we have $\Tr(\YM^{\top}\RM\RM^{\top}\YM)=\Tr(\YM\YM^{\top}\RM\RM^{\top})$ and $\Tr(\FY^{\top}\RM\RM^{\top}\FY)=\Tr(\FY\FY^{\top}\RM\RM^{\top})$. Consequently,
\begin{normalsize}
\begin{align}
& \left|\|\RM-\YM\HM^{\top}\|^2_F-\|\RM-\FY\HM^{\top}\|^2_F \right|\\
&=\left|\Tr(\YM^{\top}\RM\RM^{\top}\YM)-\Tr(\FY^{\top}\RM\RM^{\top}\FY) \right|\notag\\
& = \left|\Tr((\YM\YM^{\top}-\FY\FY^{\top})\RM\RM^{\top}) \right|. \notag
\end{align}
\end{normalsize} 
The lemma is therefore proved.
\end{proof}

\begin{proof}[\bf Proof of Lemma \ref{lem:FF-ZZ}]
First, define $\PM_\U$ as follows:
\begin{normalsize}
\begin{equation}\label{eq:P}
\PM_\U = (1-\alpha)\sum_{r=0}^{\infty}{\alpha^r \cdot  \left(\LM_\U\LM_\U^{\top}\right)^{r}}.
\end{equation}
\end{normalsize}
Suppose that $\GaM\SGVM\PsiM^{\top}$ be the exact top-$d$ SVD of $\XM_\U$, by Eckart–Young Theorem \cite{gloub1996matrix} (Theorem \ref{lem:eym}), we have $\|\GaM\SGVM\PsiM^{\top}-\XM_\U\|_2\le \sigma_{d+1}$,
where $\sigma_{d+1}$ is the $(d+1)$-th largest singular value of $\XM_\U$. Additionally, we can obtain $\|\GaM\SGVM^2\GaM^{\top}-\XM_\U\XM_\U^{\top}\|_2 \le \sigma_{d+1}^2$.
By the definition of $\PM_\U$ in Eq. \eqref{eq:P}, we have 
$$\PM_\U = \DM^{1/2}_{\U}\cdot \sum_{r=0}^{\gamma}{\alpha^t\cdot  \left(\DM^{-1}_{\U}\BM\DM^{-1}_{\V}\BM^{\top}\right)^{r} \cdot \DM^{-1/2}_{\U}},$$
where $\DM^{-1}_{\U}\BM$ and $\DM^{-1}_{\V}\BM^{\top}$ are two row-stochastic matrices, i.e., the entries at each row sum up to $1$. Let $\MM=\DM^{-1}_{\U}\BM\DM^{-1}_{\V}\BM^{\top}$. Since the multiplication of two row-stochastic matrices yields a row-stochastic matrix, $\MM=\DM^{-1}_{\U}\BM\DM^{-1}_{\V}\BM^{\top}$ is a row-stochastic matrix, which further connotes that $\MM^r = \left(\DM^{-1}_{\U}\BM\DM^{-1}_{\V}\BM^{\top}\right)^{r}$ is row-stochastic, i.e., $\|\MM[i]\|_1=\sum_{u_\ell\in \U}{\MM[i,\ell]}=1\ \forall{u_i\in \U}$. Hence, given a matrix $\boldsymbol{\Pi}=\sum_{r=0}^{\infty}{\alpha^r \MM^r}$, 
$$\|\boldsymbol{\Pi}[i]\|_1=\sum_{r=0}^{\infty}{\alpha^r}=\frac{1}{1-\alpha}\ \forall{u_i\in \U}.$$
Therefore, we can derive that
\begin{normalsize}
\begin{align*}
& \left|{\FM[i]\cdot \FM[j]}-{\ZM_\U[i]\cdot \ZM_\U[j]}\right| \\
& = \left|{(\PM_\U\GaM\SGVM)[i]\cdot (\PM_\U\GaM\SGVM)[j]}-{(\PM_\U\XM_\U)[i]\cdot (\PM_\U\XM_\U)[j]}\right|\\
& = \left| \PM_\U[i]\cdot (\GaM\SGVM^2\GaM^{\top}-\XM_\U\XM_\U^{\top}) \cdot \PM_\U[j]^{\top} \right| \\
& = \sum_{u_\ell\in \U}{\PM_\U[i,\ell]\cdot \sum_{u_h\in \U}{\PM_\U[j,h]\cdot \left(\GaM\SGVM^2\GaM^{\top}-\XM_\U\XM_\U^{\top}\right)[\ell,h]}}\\
& \le \sum_{u_\ell\in \U}{\PM_\U[i,\ell]\cdot \sum_{u_h\in \U}{\PM_\U[j,h]\cdot \sigma_{d+1}^2 }}\\
& = \sum_{u_\ell\in \U}{ \sqrt{\frac{\DM_\U[i,i]}{\DM_\U[\ell,\ell]}}\cdot \boldsymbol{\Pi}[i,\ell] \cdot \sum_{u_h\in \U}{ \sqrt{\frac{\DM_\U[j,j]}{\DM_\U[h,h]}}\cdot \boldsymbol{\Pi}[j,h] \cdot \sigma_{d+1}^2 }}\\
&\le \sqrt{\DM_\U[i,i]\cdot \DM_\U[j,j]}\cdot \frac{\sigma_{d+1}^2}{1-\alpha},
\end{align*}
\end{normalsize}
which finishes the proof.
\end{proof}

\balance

\bibliographystyle{ACM-Reference-Format}
\bibliography{main}

\end{document}